\documentclass[11pt]{article}
\usepackage[margin=1in]{geometry}

\usepackage{amsmath, amssymb, amsthm}
\usepackage[mathic]{mathtools}
\usepackage{algorithmic}
\usepackage{algorithm}
\usepackage[T1]{fontenc}
\usepackage{libertine}
\usepackage{tikz}
\usepackage{booktabs}
\usetikzlibrary{shapes}
\usetikzlibrary{arrows, positioning}
\usepackage{float}
\usepackage{comment}
\usepackage[libertine]{newtxmath} 
\usepackage[scaled=0.96]{zi4} 
\usepackage{a4wide, dsfont, microtype, xcolor, paralist}

\usepackage{multirow}
\usepackage[ocgcolorlinks]{hyperref} 
\colorlet{DarkRed}{red!50!black}
\colorlet{DarkGreen}{green!50!black}
\colorlet{DarkBlue}{blue!50!black}
\hypersetup{
	linkcolor = DarkRed,
	citecolor = DarkGreen,
	urlcolor = DarkBlue,
	bookmarksnumbered = true,
	linktocpage = true,
	pdfauthor=author
}

\usepackage{amsmath, amssymb, amsthm}
\usepackage[capitalise]{cleveref}
\usepackage{dsfont, makecell, multirow}
\usepackage{thm-restate}
\usepackage{authblk}

\DeclareMathOperator{\fim}{FIM}

\DeclareMathOperator{\al}{AL}

\DeclareMathOperator{\ssat}{SAT}
\DeclareMathOperator{\g}{g}

\DeclareMathOperator{\np}{NP}
\DeclareMathOperator{\pp}{P}

\DeclareMathOperator{\var}{var}
\DeclareMathOperator{\XP}{XP}
\DeclareMathOperator{\FPT}{FPT}

\newcommand{\NP}{\ensuremath{\np}\xspace}
\newcommand{\PP}{\ensuremath{\pp}\xspace}

\newcommand{\FIMAL}{\ensuremath{\fim_{\al}}\xspace}
\newcommand{\FIMALG}{\ensuremath{\fim_{\al}^{\g}}\xspace}
\newcommand{\STWOSAT}{\ensuremath{\Sigma_2\ssat}\xspace}

\DeclareMathOperator{\opt}{opt}
\DeclareMathOperator{\argmin}{argmin}
\DeclareMathOperator{\argmax}{argmax}
\DeclareMathOperator{\poly}{poly}

\DeclareMathOperator{\E}{\mathbb{E}}

\let\epsilon\varepsilon
\let\eps\varepsilon

\newcommand{\ones}{\mathds{1}}

\newcommand{\UUU}{\mathcal{U}}

\newcommand{\CCC}{\mathcal{C}}

\newcommand{\SSS}{\mathcal{S}}
 
\newcommand{\III}{\mathcal{I}}
\newcommand{\DDD}{\mathcal{D}}
\newcommand{\MMM}{\mathcal{M}}

\newcommand{\algo}[1]{\texttt{\textbf{#1}}}

\usepackage{xspace}

\usepackage{tikz}
\usetikzlibrary{decorations.pathreplacing}
\usetikzlibrary{plotmarks}
\usetikzlibrary{positioning,automata,arrows}
\usetikzlibrary{shapes.geometric}
\usetikzlibrary{decorations.markings}
\usetikzlibrary{positioning, shapes, arrows}

\PassOptionsToPackage{usenames,dvipsnames,svgnames}{xcolor}
\definecolor{orange}{RGB}{235,90,0}
\definecolor{darkorange}{RGB}{175,30,0}
\definecolor{turkis}{RGB}{131,182,182}
\definecolor{darkturkis}{RGB}{31,82,82}
\definecolor{green}{RGB}{102,180,0}
\definecolor{darkgreen}{RGB}{51,90,0}
\definecolor{myblue}{RGB}{0,0,213}
\definecolor{mydarkblue}{RGB}{0,0,100}
\definecolor{mybrightblue}{HTML}{74B0E4}
\definecolor{mybrighterblue}{HTML}{B3EAFA}
\definecolor{lila}{RGB}{102,0,102}
\definecolor{darkred}{RGB}{139,0,0}
\definecolor{darkyellow}{RGB}{188,135,2}
\definecolor{brightgray}{RGB}{200,200,200}
\definecolor{darkgray}{RGB}{50,50,50}
\definecolor{amaranth}{rgb}{0.9, 0.17, 0.31}
\definecolor{alizarin}{rgb}{0.82, 0.1, 0.26}
\definecolor{amber}{rgb}{1.0, 0.75, 0.0}
\definecolor{green(ryb)}{rgb}{0.4, 0.69, 0.2}
\definecolor{hanblue}{rgb}{0.27, 0.42, 0.81}
\definecolor{grannysmithapple}{rgb}{0.66, 0.89, 0.63}
\definecolor{darkcerulean}{rgb}{0.03, 0.27, 0.49}
\definecolor{arsenic}{rgb}{0.23, 0.27, 0.29}

\newtheorem{theorem}{Theorem}[section]
\newtheorem{lemma}[theorem]{Lemma}
\newtheorem{definition}[theorem]{Definition}
\newtheorem{corollary}[theorem]{Corollary}
\newtheorem{observation}[theorem]{Observation}

\title{Improving Fairness in Information Exposure by Adding Links}
\date{}

\begin{document}
\author[1]{Ruben Becker}
\author[2]{Gianlorenzo D'Angelo}
\author[2]{Sajjad Ghobadi}
\affil[1]{\normalsize Ca' Foscari University of Venice, Italy}
\affil[2]{Gran Sasso Science Institute, L'Aquila, Italy}

\maketitle

\begin{abstract}
	Fairness in influence maximization has been a very active research topic recently. Most works in this context study the question of how to find seeding strategies (deterministic or probabilistic) such that nodes or communities in the network get their fair share of coverage. Different fairness criteria have been used in this context. All these works assume that the entity that is spreading the information has an inherent interest in spreading the information fairly, otherwise why would they want to use the developed fair algorithms? This assumption may however be flawed in reality -- the spreading entity may be purely \emph{efficiency-oriented}. In this paper we propose to study two optimization problems with the goal to modify the network structure by adding links in such a way that efficiency-oriented information spreading becomes \emph{automatically fair}. We study the proposed optimization problems both from a theoretical and experimental perspective, that is, we give several hardness and hardness of approximation results, provide efficient algorithms for some special cases, and more importantly provide heuristics for solving one of the problems in practice. In our experimental study we then first compare the proposed heuristics against each other and establish the most successful one. In a second experiment, we then show that our approach can be very successful in practice. That is, we show that already after adding a few edges to the networks the greedy algorithm that purely maximizes spread surpasses all fairness-tailored algorithms in terms of ex-post fairness. Maybe surprisingly, we even show that our approach achieves ex-post fairness values that are comparable or even better than the ex-ante fairness values of the currently most efficient algorithms that optimize ex-ante fairness.
\end{abstract}

\section{Introduction}
The question of how information spreads through networks has been studied in various research disciplines. In computer science, the so-called influence maximization (IM) paradigm has attracted a lot of attention in the last two decades. The IM problem can be stated as follows. Given a social network $G=(V,E)$ in which information spreads according to some probabilistic model, target a set of at most $k$ seed nodes $S\subseteq V$ in such a way that $\sigma(S)$, the expect number of nodes that receive the information, is maximized~\cite{KempeKT15}. 

As online social networks play an essential role in how we acquire information nowadays and as access to information has an important impact on our lives, more recently researchers have started to study the IM framework in the presence of fairness concerns as well. Fairness may be understood w.r.t.\ individuals or communities, the former being the special case of the latter with singleton communities. Generally, in such works, for a community $C\subseteq V$, one considers the average probability $\sigma_C(S)$ of nodes in $C$ to be reached from $S$, also called the \emph{community coverage} of $C$. Then the concern is to choose $S$ in such a way that some fairness criteria on the communities is maximized. The probably most commonly used one is the maxmin (or maximin) criterion. In the most basic setting, studied, e.g., by Fish et al.~\cite{Fish19}, the goal is to find a seed set $S\subseteq V$ of size at most $k$ such that the minimum probability that nodes are reached $\min_{v\in V} \sigma_v(S)$ is maximized.

Several articles have been published in this scope and they have shown certain success in finding seeding strategies that lead to fairer outcomes. Nevertheless, they are all based on the assumption that the information spreading entity, i.e., the agent choosing $S$, has an interest in spreading information in a fair way. This assumption is however rather unrealistic in the real word. Information spreading agents may be, and probably mostly are, purely efficiency-oriented and not particularly interested in choosing fair seeding strategies. 

In this work we take a different approach to fairness. We do not rely on the good will of the information spreading entity, but instead modify the underlying social network in such a way as to make efficiency-oriented information spreading automatically fair. The modification of the network may be done by the network owner or any other entity interested in guaranteeing fairness. While different ways of modifying the network are perceivable, we choose the possibly most natural one -- we improve the network's connectivity by adding links. Here, we take the rather realistic approach to assume the information spreading entity to be indifferent rather than adversarial towards fairness.

\paragraph{Our Contribution.}
We formalize this problem as follows. Given a social network $G=(V,E)$, we want to add at most $b$ non-edges $F\subseteq \bar E = V^2\setminus E$ to $G$ in such a way that the minimum community coverage is maximized when information is spread in $G'=(V, E\cup F)$ using a purely efficiency oriented seeding strategy, i.e., a seed set $S$ of size $k$ that maximizes the spread in $G'$, we measure this using the function $\sigma(\cdot, F)$. We call this the \FIMAL problem -- fair influence maximization by adding links. We study the complexity of solving \FIMAL (\cref{sec: fimal}) and provide plenty of evidence that solving \FIMAL is challenging, both exactly and approximately. Maybe most importantly, we show that it is unlikely to be able to find an $\alpha$-approximation to the optimal solution, for any $\alpha\in (0,1]$, even when having access to an oracle that solves an \NP-complete problem. We furthermore show that \FIMAL remains \NP-hard for constant $b$ or $k$ (in the latter case even to be approximated).

We thus turn to study a second problem (\cref{sec: fimalg}) that is possibly practically better motivated in the first place -- the \FIMALG problem: Here instead of assuming that the efficiency oriented entity uses maximizing sets to spread information, we assume it to employ the greedy algorithm. This is a quite realistic assumption as the problem of finding a maximizing set is \NP-hard, while the greedy algorithm can be used in order to obtain a $1 - 1/e - \eps$-approximation for any $\eps\in (0,1)$ w.h.p.\ in $\poly(n, \eps^{-1})$ time, i.e, polynomial time in $n=|V|$ and $\eps^{-1}$. Even more, this approximation guarantee is essentially optimal~\cite{KempeKT15}. Multiple implementations of the greedy algorithm for IM exist (e.g., \cite{TangXS14, TangSX15}) and they have been shown to be extremely efficient in practice. We observe that, in contrast to \FIMAL, the \FIMALG problem is polynomial time solvable when $b$ is a constant -- exactly in the (unrealistic) case of deterministic instances and up to an arbitrarily small additive error in the probabilistic case. While this highlights the difference between the two problems, the proposed algorithm is essentially a brute-force algorithm and is thus not promising in practice. We complement the finding of this algorithm for the special case of constant $k$ with a lower bound showing that it is \NP-hard to provide any approximation algorithm. We then propose a set of algorithms for \FIMALG and evaluate them against each other in a first experiment in~\cref{sec: experiments}. We then take the best performing algorithm for \FIMALG and, in a second experiment, compare the resulting fairness (i.e., fairness achieved by the greedy algorithm after adding the proposed non-edges to the graph) with competitor algorithms that choose seeds as to optimize fairness. We observe that already after adding very few edges to graphs with thousands of nodes, the fairness achieved by our algorithm outperforms the fairness achieved by the fairness-tailored algorithms. Maybe surprisingly, this even holds for algorithms that optimize ex-ante fairness.

We summarize our theoretical results for \FIMAL and \FIMALG in Table~\ref{tbl:bounds} together with references to the respective statements in later sections.
\begin{table}[ht]
    \centering{{\small
        \begin{tabular}{c||c|c|c|}
                    & general                                   & constant $b$                          & constant $k$                              \\ \hline\hline
            \FIMAL  & \makecell{$\Sigma_2^p$-hard\\
                        {\small
                        $[$Thm.~\ref{thm: dFIMAE sigma2p}$]$}\\
                        $\Sigma_2^p$-hard to \\
                        $\alpha$-approx.\\
                        {\small
                        $[$Thm.~\ref{thm: fimal inapx}$]$}}  & \makecell{\NP-hard\\
                                                                {\small
                                                                $[$Thm.~\ref{thm: FIMAE inapx b1}$]$}}
                                                                                                        & \makecell{\NP-hard to\\ 
                                                                                                            $\alpha$-approx.\\
                                                                                                            {\small
                                                                                                            $[$Thm.~\ref{thm: FIMAE inapx k1}$]$}}  \\ \hline
            \FIMALG & \makecell{\NP-hard to\\ 
                         $\alpha$-approx.\\
                        {\small
                        $[$Cor.~\ref{cor: FIMAEg inapx}$]$}}  & \makecell{
                                                                - poly.\ time\\ 
                                                                (determ.)\\
                                                                {\small
                                                                $[$Obs.~\ref{obs: FIMALG det}$]$}\\
                                                                - $\eps$-approx.\\
                                                                (prob.)\\
                                                                 {\small
                                                                $[$Lem.~\ref{lem: FIMALG apx}$]$}}      & \makecell{\NP-hard to\\ 
                                                                                                            $\alpha$-approx.\\
                                                                                                        {\small
                                                                                                        $[$Cor.~\ref{cor: FIMAEg inapx}$]$}}        \\ \hline
        \end{tabular}}
    }
\caption{Summary of our complexity results. The number $\alpha$ can be any factor in $(0,1]$. }
\label{tbl:bounds}
\end{table}

\paragraph{Related Work.}
There is a rich set of related works in this area, so we are able to summarize only the results most related to our work. Fish et al.~\cite{Fish19} were the first to study the maximin criterion in influence maximization w.r.t.\ single nodes. Tsang et al.~\cite{TsangWRTZ19} study the maximin criterion w.r.t.\ groups. Becker et al.~\cite{BeckerDGG22} also consider the maximin criterion for groups, but allow probabilistic seeding strategies. Stoica and Chaintreau~\cite{StoicaC19} analyze the fairness achieved by standard algorithms for influence maximization. 

There are several works in which the authors add links to the network, however they do so with a different objective. Both Castiglioni, Ferraioli,
and Gatti~\cite{CastiglioniF020} and Cor{\`{o}}, D’Angelo, and Velaj~\cite{CoroDV21} study the problem of adding edges to a graph in order to maximize the influence from a given seed set in different models of diffusion.
Castiglioni et al.~\cite{CastiglioniF020} prove that, for the independent cascade model, 
it is $\NP$-hard to approximate the problem to within any constant factor. Cor{\`{o}} et al.~\cite{CoroDV21} study the problem with the goal of adding a limited number of edges to a given set of seed nodes. They considered the linear threshold model and proposed a constant approximation algorithm.
D’Angelo, Severini, and Velaj~\cite{DAngeloSV19} study the problem of adding a set of edges \emph{incident to a given seed set} with the same aim. In a setting, where the cost of adding each edge is 1, the authors showed that it is NP-hard to approximate the problem within a factor better than $1-1/(2e)$, and they proposed an algorithm with an approximation factor of $1-1/e$ for the independent cascade model. They extended the results to the general case where the cost of adding each edge is in $[0, 1]$.
Wu, Sheldon, and Zilberstein~\cite{wu2015efficient} consider also different intervention actions than just adding edges, e.g., increasing the weights of edges. The authors show that, for the independent cascade model, the problem of maximizing spread under these interventions is NP-hard and the objective function is neither submodular nor supermodular.
Khalil, Dilkina, and Song~\cite{KhalilDS14} study both the edge addition and deletion problems in order to maximize/minimize influence in the linear threshold model. They showed that the objective functions of both problems are supermodular and therefore there are algorithms for the problems with provable approximation guarantees.

Swift et al.~\cite{FairSpread} introduce a problem to suggest a set of edges that contains at most $k$ edges incident to each node to maximize the expected number of reached nodes while satisfying a fairness constraint (reaching each group in the network with the same probability). 
They show that the problem is NP-hard and even is NP-hard to approximate to within any bounded factor unless $P=NP$. Then, by violating the fairness constraint, they propose an LP-based algorithm with a factor of $1-1/e$ on the total spread and $2e/(1-1/e)$ on fairness. The main difference between our work and the problem studied in \cite{FairSpread} is that, the set of seeds in \cite{FairSpread} is fixed, known and
independent of the added edges. While our aim is to achieve fairness automatically, when
an external agent selects an efficient seed set that may explicitly depend on the added
edges.
Bashardoust et al.~\cite{BashardoustLinhFair} study the maximin criterion w.r.t. nodes where the goal is to add at
most $k$ edges to the network to maximize the minimum probability that a node receives the
information. They consider a case where each node is the source of distinct and equally-important information and information spread follows the independent cascade model with a transaction probability $\alpha \in [0, 1]$. The authors propose heuristics without providing any approximation guarantee and experimentally show that adding edges to the network can increase the minimum probability that nodes receive the information.
Garimella et al.~\cite{GarimellaMGM17} address the problem of recommending a set of edges to minimize the controversy score of the graph. The authors proposed an algorithm without providing any approximation guarantee. Moreover, they do not consider any diffusion process
in the network.
Tong et al.~\cite{TongPEFF12} transform the edge addition/deletion problem to the problem of maximizing/minimizing the eigenvalue of the adjacency matrix.
Amelkin and Singh~\cite{AmelkinS19} propose an edge recommendation algorithm to disable an attacker that aims to change the network's opinion by influencing users.



\section{Preliminaries}\label{sec: preliminaries}
For an integer $k$, we denote with $[k]$ the set of integers from $1$ to $k$. We say that an event holds with high probability (w.h.p.), if it holds with probability at least $1-n^{-\alpha}$ for a constant $\alpha$ that can be made arbitrarily large. 

\paragraph{Information Diffusion.}
Given a directed graph \(G=(V, E, w)\) with $n$ nodes $V$, edge set $E$, and edge weight function $w:V^2\rightarrow [0,1]$, we use the \emph{Independent Cascade model}~\cite{KempeKT15} for describing the random process of information diffusion. For an initial seed set \(S\subseteq V\) the spread $\sigma(S)$ from $S$ is the expected number of nodes that are reached from $S$ in a random \emph{live-edge graph} which is constructed as follows.
Every node \(v\) independently picks a \emph{triggering set} \(T_v\) by letting each $u$ in its set of in-neighbors $N_v$ be in $T_v$ independently with probability $w_e$, where $e=(u,v)$. We then let \(E_L := \bigcup_{v\in V}\{(u,v) : u\in T_v\}\) and we call $L=(V, E_L)$ a random live-edge graph. We then define \(\rho_L(S)\) as the set of nodes reachable from \(S\) in \(L\) and the expected number of nodes reached from $S$ is \(\sigma(S):=\E[|\rho_L(S)|]\), where the expectation is over the random generation of the live-edge graph $L$.
We furthermore define $\sigma_v(S) := \E[\ones_{v\in \rho_L(S)}] = \Pr[v\in \rho_L(S)]$ for every node $v\in V$, i.e., $\sigma_v(S)$ is the probability of $v$ being reached from $S$. For a set (or group) of nodes $C\subseteq V$, we let $\sigma_C(S) := \frac{1}{|C|} \sum_{v\in C} \sigma_v(S)$ be the group coverage of $C$.

When the edge probabilities belong to $\{0,1\}$, we refer to the instance as the \emph{deterministic case}, in this case $\sigma(S)$ is the (deterministic) number of nodes reachable from seeds $S$ in $G$.
In the general case, it is not feasible to compute $\sigma(S)$ via all live-edge graphs $L$, instead a \((1\pm\epsilon)\)-approximation of \(\sigma(S)\) can be obtained w.h.p.\ by averaging over \(\poly(n, \eps^{-1})\) many live-edge graphs \(L\), see, e.g., Proposition 4.1 in the work of Kempe, Kleinberg, and Tardos~\cite{KempeKT15}. Similarly, $\sigma_v(S)$ (and thus also $\sigma_C(S)$) can be approximated w.h.p., however, only to within an additive error of $\pm\eps$ by averaging over \(\poly(n, \eps^{-1})\) many live-edge graphs, see, e.g., Lemma~4.1 in the work of Becker et al.~\cite{BeckerDGG22}. 

\paragraph{Non-Edges and Spread with Added Edges.}
We let $\bar E := (V\times V) \setminus E$ denote the set of \emph{non-edges} in $G$. For a set of non-edges $F \subseteq \bar E$ and a set of seed nodes $S \subseteq V$, we define $\sigma(S, F)$ as the expected number of nodes reached from $S$ in the graph $G'=(V, E \cup F)$ that results from adding $F$ to $G$. This is the reason why we have defined the edge weight function also w.r.t.\ non-edges above. Similarly, for a node $v\in V$, $\sigma_v(S, F)$ is the probability that $v$ is reached from $S$ in $G'$ and, for a community $C \subseteq V$, we define $\sigma_C(S, F)$ to be the average probability of nodes in $C$ being reached from $S$ in $G'$. We remark that also these functions cannot be computed exactly but only approximated in the same way as their counterparts without added edges.

\section{The \texorpdfstring{\FIMAL}{FIMAL} Problem: Making Spread Maximizers Fair}\label{sec: fimal}
\paragraph{Problem Definition.}
Consider a directed weighted graph $G = (V, E, w)$ and let $\CCC$ be a community structure, i.e., $m$ non-empty communities $C \subseteq V$, and let $k$ and $b$ be two integers. For a set of non-edges $F \subseteq \bar E$, we define $\MMM(F, k) := \argmax_{S\subseteq V} \{ \sigma(S, F) :  |S| \le k \} $ to be the set of size $k$ maximizers to $\sigma(\cdot , F)$.
We are now ready to formally define the \FIMAL problem motivated above:
\begin{align*}
	\max_{F \subseteq \bar E: |F| \le b } \big\{ \tau : \min_{C\in\CCC}\sigma_C(S, F) \ge \tau,\  \forall S \in \MMM(F, k)\big\}. 
\end{align*}
We denote with $\opt_{\al}(G, \CCC, b, k)$ the optimum of~\FIMAL.
Clearly, our objective in~\FIMAL is to find a set of at most $b$ non-edges $F \subseteq \bar E$, that, when added to $G$, maximizes the minimum community coverage when information is spread in a purely ``efficiency-oriented'' way, i.e., from a set of at most $k$ seed nodes that is chosen such that the set function $\sigma(\cdot, F)$ is maximized. The motivation behind studying~\FIMAL is to, e.g., as the network owner, change the structure of a social network in such a way that an efficiency-oriented entity that wants to spread information in $G$ automatically spreads information in a more fair way. 

In what follows, we give several hardness and hardness of approximation results for \FIMAL. We start by showing that the decision version of the general \FIMAL problem is $\Sigma_2^p$-hard. We even show that it is unlikely that \FIMAL can be approximated to within any factor. We then turn to special cases of \FIMAL where either $b=1$ or $k=1$ and show that the problem remains \NP-hard also in these special cases -- for $k=1$ even hard to approximate to within any factor. 

For better comprehensibility, we first note that in the the decision version of \FIMAL, in addition to the graph \(G=(V, E)\), the communities \(\CCC\), and the integers \(b, k\), we are given a threshold \(t\) and the task is to decide if there exists \(F \subseteq \bar E\) with \(|F| \le b\) such that for all $S\in \MMM(F, k)$: \(\min_{C\in \CCC} \sigma_C(S, F)\ge t\). 

\paragraph{$\Sigma_2^p$-Hardness.}
We start by recalling the definition of the complexity class $\Sigma_2^p$.
\begin{definition}[Definition 5.1 in~\cite{AroraBarak09}]
    The class $\Sigma_2^p$ is defined to be the set of all languages $L$ for which there exists a polynomial-time Turing machine $M$ and a polynomial $q$ such that $x\in L$ if and only if $\exists u\in \{0, 1\}^{q(|x|)}:\forall v\in \{0, 1\}^{q(|x|)}: M(x,u,v)=1$.\footnote{Equivalently, see, e.g., Theorem~5.12 and Remark~5.16 in the same book, $\Sigma_2^p$ can be defined as the set of all languages that can be decided by a non-deterministic Turing machine with access to an oracle that solves some \NP-complete problem.}
\end{definition}

We next introduce the $\STWOSAT$ problem which is $\Sigma_2^p$-complete, see, e.g., Exercise~1 in Chapter~5 of the book by Arora and Barak~\cite{AroraBarak09}. 
\begin{definition}[Example 5.6 in~\cite{AroraBarak09}]
    Given a boolean expression $\phi(X, Y)$ in 3-CNF with variables $X=(x_1, \ldots, x_\nu)$ and $Y=(y_{\nu + 1}, \ldots, y_\mu)$, the $\STWOSAT$ problem entails to decide if 
    $\exists x \forall y : \phi(x, y) = \top$, where 
    $x : X \rightarrow \{0, 1\}$ and $y : Y \rightarrow \{0, 1\}$ are assignments to the variables $X$ and $Y$, respectively.
\end{definition}
For ease of presentation, we assume the indices of $Y$ to start at $\nu + 1$, such that indices of $X$ and $Y$ are disjoint.
Our goal now is to show that the decision version of \FIMAL is $\Sigma_2^p$-hard. 
We will describe a reduction from \STWOSAT to the decision version of \FIMAL. We assume that $\phi(X, Y)$ contains $m$ clauses $\phi_1, \ldots, \phi_m$ and for a clause $\phi_r$ we call $r(s)$, $s\in [3]$, the indices of the three variables corresponding to $\phi_r$'s three literals (in arbitrary fixed order). 

Given an instance of \STWOSAT, we create an instance $(G, \CCC, b, k, t)$ of the decision version of \FIMAL as follows, see Figure~\ref{fig: sigma2preduction} for an illustration. Fix a constant $M:=\mu + \nu + 6m + 1$. The node set $V$ of $G$ consists of 
\begin{itemize}
    \item $U = \{q, P\}$, where $P =p_1, \ldots, p_{M-1}$,  
    \item $V^\exists = \{v_i, \bar v_i :  i\in [\nu] \}$ and $V^\forall = \{v_j, \bar v_j, L_j :  j\in [\mu]\setminus [\nu]\}$, where $L_j = l_{j,1}, \ldots, l_{j,M-2}$,
    and
    \item $W = \{w^r_1, w^r_{\bar 1}, w^r_2, w^r_{\bar 2}, w^r_3, w^r_{\bar 3}:  r\in [m]\}$.
\end{itemize}
The edge set $E$ consists of 
\begin{itemize}
    \item $E^{\var} := \{ (v_{r(s)}, w^r_s), (\bar v_{r(s)}, w^r_{\bar s}) : s\in [3], r\in [m]\}$,
    \item $E^L$ that consists of all edges from the nodes $v_j$, $\bar v_j$ to all nodes $v\in L_j$, for $j\in[\mu]\setminus [\nu]$,
    \item $E^P$ that consists of edges from $q$ to all nodes in $P$, and
    \item $Z:=V^2\setminus (E^{\var} \cup E^L \cup E^P \cup E(q,  V^\exists))$, where $E(q, V^\exists) := \{(q,v) : v\in V^\exists \}$.
\end{itemize}
We note that as a result $\bar E = E(q, V^\exists)$.
The edge weight function is defined as $w_e = 0$ for all edges $e\in Z$ and $w_e=1$ otherwise. 
The community structure $\CCC$
consists of: (1) communities $C_{1}, \ldots, C_{m}$, where each $C_r$ is of cardinality 3 and for $s\in [3]$, $w^r_s\in C_{r}$ if $x_{r(s)}\in \phi_r$ (or $y_{r(s)}\in \phi_r$) and $w^r_{\bar s}\in C_{r}$ if $\bar x_{r(s)}\in \phi_r$ (or $\bar y_{r(s)}\in \phi_r$); and (2) communities $C_{m+1}, \ldots, C_{m+\nu}$, with $C_{m+i}= \{v_i, \bar v_i\}$ for each $i \in [\nu]$.
We set $k=\mu+1$, $b=\nu$ and $t=1/3$. 
\begin{figure}[ht]
	\centering
	\resizebox{0.95\columnwidth}{!}{
	    \begin{tikzpicture}[scale=1.1]
			\tikzset{vertex/.style = {shape=circle,draw = black,thick,fill = white, minimum size=1cm}}
			\tikzset{edge/.style = {->,> = latex'}}
			
			\node[vertex] (q) at  (-5.5,-3.55) {$q$};

			\path [draw = black, rounded corners, inner sep=100pt]
			(-6.5, -5.1) -- (-4.5, -5.1) -- (-4.5, -5.7) -- (-6.5, -5.7) -- cycle;
			\node (P) at  (-5.5, -5.2) {};
			\node (Pl) at (-5.5, -5.4) {$P$}; 

			\draw[edge] (q) to[left] node {} (P);
			
			\path [draw = black, rounded corners, inner sep=100pt, dashed]
			(-6.7, -2.9) -- (-4.3, -2.9) -- (-4.3, -5.9) -- (-6.7, -5.9) -- cycle;
			\node  (empty) at (-6.5, -2.65)  {\Large $U$};

			\node[vertex] (vfa1) at  (-2.7, 1.2) {$ v_1$};
			\node[vertex] (vfanu) at  (-2.7,-1.2){$ v_\nu$};
			
			\node[vertex] (vfa1n)  at  (-.5, 1.2) {$\bar  v_1$};
			\node[vertex] (vfanun) at  (-.5,-1.2){$ \bar  v_\nu$};
			
			\node (d1) at (-1.5, -.1) {$\vdots$};

			\path [draw = black, rounded corners, inner sep=100pt, dotted]
			(-3.3, 1.7) -- (.1, 1.7) -- (.1, 0.4) -- (-3.3, 0.4) -- cycle;
			\node  (empty) at (-2.8, .6)  {\large $C_{m+1}$};
			
			\path [draw = black, rounded corners, inner sep=100pt, dotted]
			(-3.3, -0.7) -- (.1, -0.7) -- (.1, -2) -- (-3.3, -2) -- cycle;
			\node  (empty) at (-2.8, -1.8)  {\large $C_{m+\nu}$};
			
			\path [draw = black, rounded corners, inner sep=100pt, dashed]
			(-3.5, 1.9) -- (.3, 1.9) -- (.3, -2.2) -- (-3.5, -2.2) -- cycle;
			\node  (empty) at (-3.9, .4)  {\Large $V^{\exists}$};

			\node[vertex] (ve1) at  (-2.7,-3.5) {$v_{\nu+1}$};
			\node[vertex] (ve1n) at  (-.5,-3.5) {$\bar v_{\nu+1}$};
			\path [draw = black, rounded corners, inner sep=100pt]
			(-2.6, -4.3) -- (-.6, -4.3) -- (-.6, -4.9) -- (-2.6, -4.9) -- cycle;
			\node (L1) at  (-1.6, -4.4) {};
			\node (L1l) at  (-1.6, -4.6) {$L_{\nu + 1}$};

			\node[vertex] (vemu) at  (-2.7, -6.7){$v_\mu$};
			\node[vertex] (vemun) at  (-.5, -6.7){$ \bar v_\mu$};
			\path [draw = black, rounded corners, inner sep=100pt]
			(-2.6, -7.5) -- (-.6, -7.5) -- (-.6, -8.1) -- (-2.6, -8.1) -- cycle;
			\node (Lmu) at  (-1.6,-7.6) {};
			\node (Lmul) at  (-1.6,-7.8) {$L_\mu$};
            \node () at (-1.5, -5.3) {$\vdots$};
            
			\draw[edge] (ve1) to[left] node   {} (L1);
			\draw[edge] (ve1n) to[right] node   {} (L1);

			\draw[edge] (vemu) to[left] node {} (Lmu);
			\draw[edge] (vemun) to[right] node {} (Lmu);
			
			\path [draw = black, rounded corners, inner sep=100pt, dashed]
			(-3.5, -2.8) -- (.3, -2.8) -- (.3, -8.4) -- (-3.5, -8.4) -- cycle;
			\node  (empty) at (-3.9, -7.6)  {\Large $V^{\forall}$};

			\node[vertex] (u11)  at  (3,-1) {$w^1_{1}$};
			\node[vertex] (u1n2) at  (4,-1) {$w^1_{\bar 2}$};
			\node[vertex] (u13)  at  (5,-1) {$w^1_{3}$};
			\node[vertex] (u1n1) at  (6,-1) {$w^1_{\bar 1}$};
			\node[vertex] (u12)  at  (7,-1) {$w^1_2$};
			\node[vertex] (u1n3) at  (8,-1) {$w^1_{\bar 3}$};

			\node[vertex] (um1)at (3,-6.5) {$w^m_1$};
			\node[vertex] (um1)at (4,-6.5) {$w^m_{2}$};
			\node[vertex] (um1)at (5,-6.5) {$w^m_{\bar 3}$};
			\node[vertex] (um1)at (6,-6.5) {$w^m_{\bar 1}$};
			\node[vertex] (um1)at (7,-6.5) {$w^m_{\bar 2}$};
			\node[vertex] (um1)at (8,-6.5) {$w^m_{3}$};
			
			\path [draw = black, rounded corners, inner sep=100pt, dotted]
			(2.24, -0.5) -- (5.5, -0.5) -- (5.5, -1.5) -- (2.24, -1.5) -- cycle;
			\node  (empty)  at (2.45, -1.3)  {\large $C_1$};

			\path [draw = black, rounded corners, inner sep=100pt, dotted]
			(2.19, -6) -- (5.5, -6) -- (5.5, -7) -- (2.19, -7) -- cycle;
			\node  (empty) at (2.45, -6.8)  {\large $C_m$};
			
			\path [draw = black, rounded corners, inner sep=100pt, dashed]
			(2.1, -0.3) -- (8.6, -0.3) -- (8.6, -7.2) -- (2.1, -7.2) -- cycle;
			\node  (empty) at (8.2, -7.6)  {\Large $W$};
            \node () at (5.6, -3.5) {$\vdots$};
			
			\draw[edge] (vfa1) to[above] node   {} (u11);
			\draw[edge] (ve1n) to[below] node   {} (u1n2);
			\draw[edge] (vemu) to[above] node   {} (u13);

			\draw[edge] (vfa1n) to[above] node   {} (u1n1);
            \draw[edge] (ve1)   to[below] node   {} (u12);
			\draw[edge] (vemun) to[above] node   {} (u1n3);

	    \end{tikzpicture}
	}
	\caption{Construction of $G$ from a $\STWOSAT$ instance. Only the edges to the nodes corresponding to the first clause $\phi_1$ are drawn. All drawn edges have weight $1$. The only edges that are not in $G$ are the ones from $q$ to $V^\exists$.}\label{fig: sigma2preduction}
\end{figure}
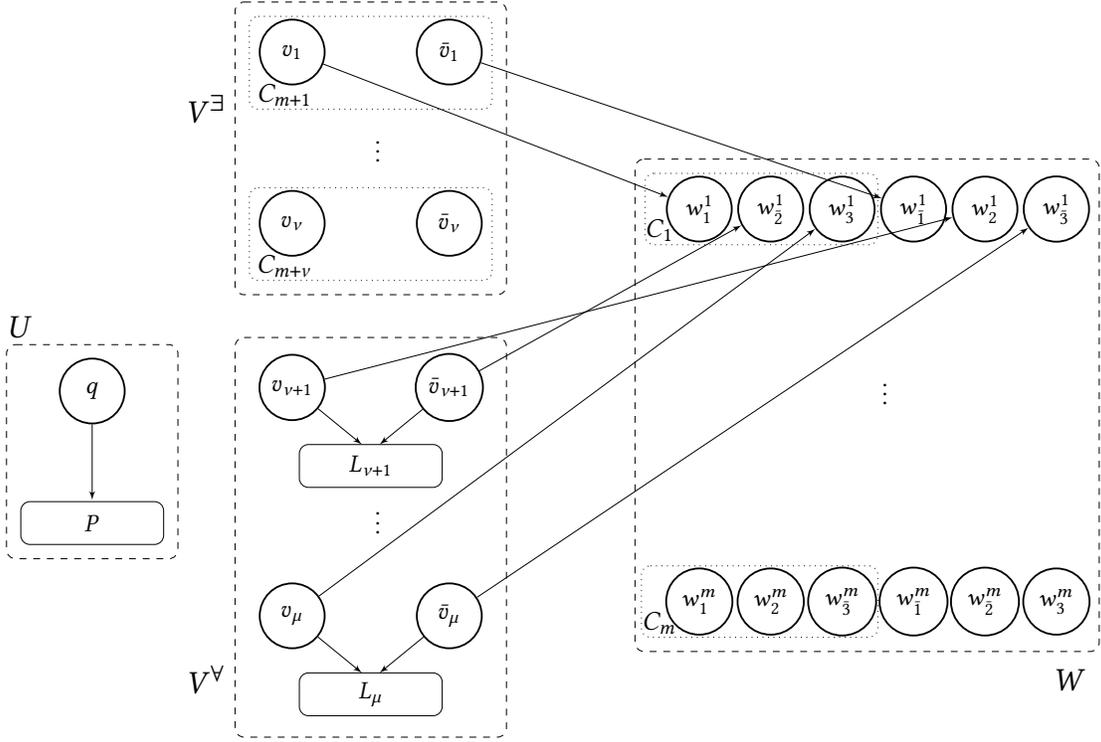

Our goal is now to show that the \STWOSAT instance is a yes-instance if and only if the constructed \FIMAL instance is. We first need the following lemma.

\begin{restatable}{lemma}{lemmasmax}\label{lemma-Smax}
	Let $F \subseteq \bar E = E(q, V^\exists)$ with $|F| \le \nu$. 
	It holds that $S\in \MMM(F, \mu + 1)$ if and only if $q\in S$ and $S\cap \{v_j, \bar v_j\}\neq \emptyset$ for all $j \in [\mu]\setminus [\nu]$. 
\end{restatable}
\begin{proof}
    Fix a set $F$ as in the statement of the lemma and let us call $P(S)$ for the property that $q\in S$ and $S\cap \{v_j, \bar v_j\}\neq \emptyset$ for all $j \in [\mu]\setminus [\nu]$.
    ($\Rightarrow$)~First note that any set $S$ that satisfies $P(S)$, achieves 
    $
        \sigma(S, F)
        \ge M + (M-1)(\mu - \nu)
    $
    and that a set $T$ that does not satisfy $P(T)$ achieves $\sigma(T, F)\le n - M$. Now, notice that $n= M + 2\nu + (\mu - \nu) M + 6m$ and thus $\sigma(T, F)\le 2\nu + (\mu - \nu) M + 6m$. Using that $M=\mu + \nu + 6m + 1$, shows that $\sigma(S, F)> \sigma(T, F)$. This shows that $T$ cannot be in $\MMM(F, k)$ and thus this completes the proof of this direction.
    ($\Leftarrow$)~It is enough to show that all sets that satisfy property $P(S)$ achieve the same value $\sigma(S, F)$.  From the construction of $E^{\var}$ it follows that the set $W$ can be partitioned into $W^\forall$ and $W^\exists$ in a way that the nodes in $W^\forall$ have an in-edge from a node in $V^\forall$, while the nodes in $W^\exists$ have an in-edge from $V^\exists$. Now, let $S$ be an arbitrary set satisfying property $P(S)$. It then follows that \[
        \sigma(S, F) = \sigma(\{q\}, F) + \frac{|W^\forall|}{2} + (M - 1)(\mu-\nu). 
    \]
    As the latter does not depend on $S$ the proof is complete.
\end{proof}
We are now ready to prove the theorem.
\begin{theorem}\label{thm: dFIMAE sigma2p}
    The decision version of \FIMAL is $\Sigma_2^p$-hard even in the deterministic case.
\end{theorem}
\begin{proof}

    We show that the \STWOSAT instance is a yes-instance if and only if the constructed \FIMAL instance is.

    ($\Rightarrow$)~Assume that the \STWOSAT instance is a yes-instance, i.e, there exists an assignments $x$ to the variables $X$ such that for all assignment $y$ to the variables $Y$, it holds that $\phi(x, y)=\top$. We will now show that there exists \(F \subseteq \bar E\) with \(|F| \le \nu\) such that for all $S\in \MMM(F, \mu + 1)$, it holds that \(\min_{C\in \CCC} \sigma_C(S, F)\ge 1/3\). Let \(F \subseteq \bar E = E(q, V^\exists)\) be equal to the set of edges from $q$ to $V^\exists$ that correspond to the assignment $x$. Now, let $S\in\MMM(F, \mu + 1)$ be arbitrary. It then follows using~\cref{lemma-Smax} that $S=\{q\} \cup \dot S$, where $\dot S$ corresponds to an assignment $y$ of $Y$. As $\phi(x, y)=\top$ it follows that, for every clause $\phi_r$ at least one literal is true, thus for every community $C_r$ with $r\in [m]$, at least one node $w\in C_r$ is reached and hence $\sigma_{C_r}(S, F)\ge 1/3$. For communities $C_i$ with $i\in [m+1, m+\nu]$, we obtain that $\sigma_{C_i}(S, F)\ge 1/2$, as $F$ corresponds to an assignment and $S$ contains $q$ according to~\cref{lemma-Smax}.
    
    ($\Leftarrow$)~Now, assume that the \FIMAL instance admits a solution \(F \subseteq \bar E\) with \(|F| \le \nu\) such that for all $S\in \MMM(F, \mu + 1)$, it holds that \(\min_{C\in \CCC} \sigma_C(S, F)> 0\). Notice that $\sigma_C(S, F)>0$ for every $S\in \MMM(F, \mu + 1)$ together with~\cref{lemma-Smax} implies that $F$ consists of a set of edges to $V^\exists$ that corresponds to an assignment. Let now $y$ be an arbitrary assignment to $Y$ and let $S$ be the set containing $q$ and all nodes from $V^\forall$ that correspond to the assignment $y$. Again using~\cref{lemma-Smax} it follows that $S\in \MMM(F, \mu + 1)$ and thus $\sigma_{C_r}(S, F) > 0$ for all $r\in [m]$. This means that at least one node in every community $C_i$ is reached or equivalently at least one literal in every clause $\phi_r$ is true in the assignments $x$ and $y$. It follows that $\phi(x,y)=\top$.
\end{proof}
From the same reduction, we can even conclude that it is unlikely to find an arbitrary approximation to \FIMAL as shown in the next theorem. 
The class $\Delta_2^p$ is the class of all languages decided by polynomial-time Turing Machines that have access to an oracle for some \NP-complete problem. It is widely believed that $\Sigma_2^p$ and $\Delta_2^p$ are distinct (see Section~17.2 in~\cite{booksChristos}).
\begin{restatable}{theorem}{fimalinapx}\label{thm: fimal inapx}
    Let $\alpha\in (0,1]$. If computing an $\alpha$-approximation to \FIMAL is in $\Delta_2^p$, then $\Sigma_2^p=\Delta_2^p$.
\end{restatable}
\begin{proof}
    Note that we have shown above that the \STWOSAT instance is a yes-instance if and only if the constructed \FIMAL instance admits a solution \(F \subseteq \bar E\) with \(|F| \le \nu\) such that for all $S\in \MMM(F, \mu + 1)$, it holds that \(\min_{C\in \CCC} \sigma_C(S, F)> 0\). Note also that the \FIMAL instance there is deterministic.
    
    Now, let $\alpha\in (0,1]$ and assume that we have an algorithm computing an $\alpha$-approximation to \FIMAL that runs in polynomial time when given access to an oracle for some \NP-complete problem, i.e., computing an $\alpha$-approximate solution to \FIMAL is in $\Delta_2^p=\PP^{\NP}$. Given a \STWOSAT instance, we can then build the \FIMAL instance as described and compute an $\alpha$-approximation to it. We then get a set \(F \subseteq \bar E\) with \(|F| \le \nu\) such that for all $S\in \MMM(F, \mu + 1)$, it holds that \(\min_{C\in \CCC} \sigma_C(S, F) \ge \alpha \cdot \opt_{\al}(G, \CCC, b, k)\). Therefore, the original \STWOSAT instance is a yes-instance if and only if \(\min_{C\in \CCC} \sigma_C(S, F)>0\), for all $S\in \MMM(F, \mu + 1)$, and, if we can check this last condition, then we can decide whether the \STWOSAT instance is a yes-instance.
    We now show how to check this condition by using a polynomial number of calls to an oracle for some $\NP$-complete problem.

    We equivalently show how to check whether there exists a solution $S\in \MMM(F, \mu + 1)$ such that \(\min_{C\in \CCC} \sigma_C(S, F) =0\). In deterministic instances, it is $\NP$-complete to check whether there exists a seed set $S$ such that $\sigma(S, F) \geq \tau$, for some parameter $\tau$. We can then, using a polynomial number of calls to the oracle, find an $S\in \MMM(F, \mu + 1)$. In fact, since the instance is deterministic, it is enough to guess all $\tau\in[|V|]$. Let now $\tau^* = \sigma(S, F)$. Then we again use an oracle to solve the $\NP$-complete problem of checking whether there exists a seed set $S$ such that $\sigma(S, F) = \tau^*$ and \(\min_{C\in \CCC} \sigma_C(S, F) =0\). As the above algorithm overall requires a polynomial number of calls to the oracle, the proof is complete.
\end{proof}

\paragraph{Still Hard Special Cases.}
While we have shown above that the general problem is $\Sigma_2^p$-hard, we will now show that not even in the apparently simple case where $k=1$, we can hope to find any approximation unless $\PP=\NP$.
\begin{restatable}{theorem}{fimaeinapxk} \label{thm: FIMAE inapx k1}
	For any, $\alpha\in (0, 1]$, it is \NP-hard to approximate \FIMAL to within a factor of $\alpha$, even in the deterministic case and if $k=1$.
\end{restatable}
\begin{proof}
    We reduce from \textsc{Set Cover}, where we are given a collection of sets $\DDD=\{D_1, \ldots, D_\mu\}$ over a ground set $\UUU = \{U_1, \ldots, U_\nu\}$ and an integer $\kappa$, and the task is to decide whether there exists a set cover of size at most $\kappa$, i.e., a collection $\SSS\subseteq \DDD$ with $|\SSS|\le \kappa$ such that $\bigcup_{D\in \SSS} D = \UUU$.

    Given a \textsc{Set Cover} instance, we create an instance $(G, \CCC, b, 1)$ of \FIMAL as follows. The graph $G=(V, E, w)$ has node set $V:=A\cup B \cup \{q\}$, where $A:=\{v_{1}, \ldots, v_{\mu}\}$ and $B:=\{u_{1}, \ldots, u_{\nu}\}$ and edge set $E:=E^{sc} \cup Z$,
    where $E^{sc}:=\{(v_j, u_i): U_i\in D_j\}$ and $Z:= V^2 \setminus (E^{sc} \cup E_{q,A})$, where $E_{q,A}:=\{q\}\times A$. The edge-weight function $w$ is defined as $w_e=1$ for $e\in E^{sc}\cup E_{q,A}$ and $w_e = 0$ otherwise, i.e., for $e\in Z$.
	The communities $\CCC$ consist of $\nu+1$ singletons $C_q= \{q\}$ and $C_i = \{u_i\}$ for $i\in[\nu]$. We set $b=\kappa$.
	
    We now show that there exists a set cover $\SSS$ of size at most $\kappa$ if and only if there exists a set of non-edges $F \subseteq \bar E$ with $|F|\le b$, such that $\min_{C\in \CCC}\sigma_C(S, F)\ge 1$ for all $S\in\MMM(F, k)$: (``$\Rightarrow$'')~Assume that there exists a set cover $\SSS$ of size at most $\kappa$. Consider the set $F=\{(q, v_j) : D_j \in \SSS\}$ that is of cardinality at most $b=\kappa$.
    We now observe that $\MMM(F, k)=\{\{q\}\}$ and thus $\min_{C\in \CCC}\sigma_C(S, F)\ge 1$ for all $S\in \MMM(F, k)$ by the choice of $F$.
	(``$\Leftarrow$'')~Now assume that there exists a set $F\subseteq \bar E$ with $|F|\le b = \kappa$ such that $\min_{C\in \CCC}\sigma_C(S, F) \ge  1$ for all $S\in\MMM(F, k)$.
    Note that $F \subseteq \bar E = E_{q, A}$ and thus again $\MMM(F, k)=\{\{q\}\}$ and $\sigma(S, F)= \nu + \kappa + 1$ for all $S\in \MMM(F, k)$. Hence, it follows that $\{D_j : (q, v_j)\in F\}$ is a set cover of size at most $\kappa$.

	Now, let $\alpha\in (0,1]$ and assume that there exists a polynomial time $\alpha$-approximation algorithm $\mathcal{A}$ for \FIMAL. We obtain that if there is a set cover of size $\kappa$, then $\opt_{\al}(G, \CCC, b, k)=1$ and $\mathcal{A}$ outputs a set $F$ such that $\min_{C\in \CCC}\sigma_C(S, F)\ge \alpha \cdot \opt_{\al}(G, \CCC, b, k) > 0$ for all sets $S\in\MMM(F, k)$.
	If however there is no set cover of size $\kappa$, then $\opt_{\al}(G, \CCC, b, k)<1$ and as the instance is deterministic this means that  $\opt_{\al}(G, \CCC, b, k)=0$. Thus $\mathcal{A}$ must return a solution $F$ such that $\sigma_C(S, F) =0$ for some community $C\in \CCC$ and some set $S\in\MMM(F, k)$. Therefore, by using $\mathcal{A}$ we can decide in polynomial time whether or not there exists a set cover of size $\kappa$ by running $\mathcal{A}$ and then checking if there exists a community $C\in \CCC$ and a set $S\in\MMM(F, k)$ such that $\sigma_C(S, F) =0$. Note that we can compute $\MMM(F, k)$ in polynomial time by evaluation of all different $n$ choices -- recall that $k=1$. It follows that it is \NP-hard to approximate \FIMAL to within a factor of $\alpha$.
\end{proof}

A natural next question is whether the problem remains hard also if $b=1$. We show that this is the case:
\begin{restatable}{theorem}{fimaeinapxb} \label{thm: FIMAE inapx b1}
	The decision version of \FIMAL is \NP-hard even in the deterministic case and if $b=1$.
\end{restatable}
\begin{proof}
    We reduce from \textsc{Set Cover}, where we are given a collection of sets $\DDD=\{D_1, \ldots, D_\mu\}$ over a ground set $\UUU = \{U_1, \ldots, U_\nu\}$ and an integer $\kappa$, and the task is to decide whether there exists a set cover of size at most $\kappa$, i.e., a collection $\SSS\subseteq \DDD$ with $|\SSS|\le \kappa$ such that $\bigcup_{D\in \SSS} D = \UUU$. W.l.o.g., we can assume that every $U_i$ appears in at least one set $D_j$ as otherwise the instance is trivially a no-instance.

	Given a \textsc{Set Cover} instance, we create an instance $(G, \CCC, 1, k, t)$ of the decision version of \FIMAL as follows (here $t$ denotes the threshold to be reached). The graph $G=(V, E, w)$ has node set $V:=A\cup B \cup \{q\}$, where $A:=\{v_{1}, \ldots, v_{\mu}\}$ and $B:=\{u_{1}, \ldots, u_{\nu}\}$ and edge set $E:=E^{sc} \cup Z$, where $E^{sc}:=\{(v_j, u_i): U_i\in D_j\}$ and $Z= V^2 \setminus (E^{sc} \cup E_{B,q})$ with $E_{B, q}:=B\times \{q\}$. The edge-weight function $w$ is defined as $w_e=1$ for $e\in E^{sc}\cup E_{B,q}$ and $w_e = 0$ otherwise, i.e., for $e\in Z$.
	The community structure $\CCC$ consists of $\nu+1$ singleton communities $C_q= \{q\}$ and $C_i = \{u_i\}$ for every $i\in[\nu]$. We set $k=\kappa$ and $t=1$.
	
    We now show that the set cover instance is a yes-instance if and only if the \FIMAL instance is, i.e., if there exists a set of non-edges $F \subseteq \bar E$ with $|F|\le b$, such that $\min_{C\in \CCC}\sigma_C(S, F)\ge 1$ for all $S\in\MMM(F, k)$: 
    (``$\Rightarrow$'')~Assume that there is a set cover $\SSS$ of size at most $\kappa$. Let $F = \{(u, q)\}$ for some arbitrary node $u \in B$. Then $S=\{v_j: D_j\in \SSS\}$ achieves $\sigma(S, F)= \nu + \kappa + 1$. Note that nodes in $A$ have no ingoing edges with positive probability and thus no set that is not a subset of $A$ can achieve a higher coverage than $S$ thus $\MMM(F, k)=\{S\}$. As a consequence $\min_{C\in \CCC} \sigma_C(S_{F, k}, F) \ge 1$ for all $S\in \MMM(F,k)$.
    (``$\Leftarrow$'')~Now assume that there exists a set $F\subseteq \bar E$ with $|F|\le b = 1$, such that $\min_{C\in \CCC}\sigma_C(S, F)\ge 1$ for all $S\in \MMM(F,k)$. Note that $F \subseteq \bar E = E_{B,q}$ and thus from $\min_{C\in \CCC}\sigma_C(S, F)\ge 1$, it follows that $\sigma_C(S, \emptyset)\ge 1$ for every $C=\{u_i\}$ and $S\in\MMM(F, k)$.
    By the assumption on the \textsc{Set Cover} instance, the set $S$ can be transformed into a subset $S'$ of $A$ such that still $\sigma_C(S', \emptyset)\ge 1$ for every $C=\{u_i\}$. We can thus conclude that $\{D_i:v_i\in S'\}$ is a set cover of size at most $\kappa$.
\end{proof}

\section{The \texorpdfstring{\FIMALG}{FIMALG} Problem: Towards Fairness in Practice}\label{sec: fimalg}
\paragraph{Problem Definition.}
We have seen a lot of evidence above that \FIMAL is intractable. We thus continue by proposing an alternative problem that not only turns out to be more computationally tractable, but also is possibly practically better motivated in the first place in the following sense:
The problem of finding a set of at most $k$ nodes that maximizes $\sigma(\cdot, F)$ is however an \NP-hard optimization problem and thus it is unrealistic to assume the entity to spread information using a maximizing set. Instead what is frequently used in practice for the computation of an efficient seed set is the greedy algorithm. In fact, the choice of the greedy algorithm is also well-founded in theory, as, for a fixed set of non-edges $F$, the set function $\sigma(\cdot, F)$ is monotone and submodular and thus one is guaranteed to achieve an essentially optimal approximation factor of $1-1/e-\eps$ for any $\eps>0$, see the work of Kempe, Kleinberg, and Tardos~\cite{KempeKT15}. Hence, an optimization problem that is practically better motivated than~\FIMAL, assumes that the efficiency-oriented entity, in order to spread information, uses the greedy algorithm for computing the seed set. The greedy algorithm for $\sigma(\cdot, F)$ is however a randomized algorithm, as it relies on simulating information spread using a polynomial number of live-edge graphs (or reverse reachable (RR) sets, depending on the implementation). It becomes thus necessary that we consider the output of the algorithm to be a distribution over seed sets of size $k$, rather than just a single set. For a set of non-edges $F \subseteq \bar E$ and an integer $k$, let us denote this distribution with $p(F,k)$. We then define the \FIMALG problem as:
\begin{align*}
    \max_{F \subseteq \bar E: |F| \le b } \big\{ \tau : \E_{S\sim p(F,k)}[\sigma_C(S, F)] \ge \tau \;\forall\, C\in\CCC\big\}. 
\end{align*}
Intuitively, our goal in the optimization problem \FIMALG is to find a set of at most $b$ non-edges $F \subseteq \bar E$, that, when added to $G$, maximizes the minimum community coverage (in expectation) when information is spread using the greedy algorithm -- a quite realistic assumption. We assume the greedy algorithm to break ties arbitrarily, but consistently.

Here, we do not assume to have access to $p(F,k)$, not even for one set $F$, as it would generally require exponential space to be encoded. Instead, we assume to have access to the greedy algorithm in an oracle fashion, i.e., for a given set $F$, we can call the greedy algorithm on $\sigma(\cdot, F)$ with budget $k$ and get a set $S$. One can then show using an easy Hoeffding bound argument, see below, that $\E_{S\sim p(F,k)}[\sigma_C(S, F)]$ can be approximated arbitrarily well w.h.p.\ for every $F$. 

It is also worth mentioning that our approach can be extended to a setting where we want to be fair w.r.t.\ multiple implementations of the greedy algorithm or even more generally to multiple implementations of multiple algorithms (different from the greedy algorithm). This can be achieved as follows. Assume that $(p_i)_{i\in[N]}$ are a priori-likelihoods of using one of $N$ different algorithms and assume $p^i(F, k)$ to reflect the probability distribution of seed sets corresponding to algorithm $i$. Then the distribution with $p_S(F, k) := \sum_{i} p_i\cdot p^i_S(F, k)$ for $S\subseteq V$ reflects the distribution over seed sets resulting from using all $N$ algorithms. The only condition here, for our algorithmic results below to keep working, is that the algorithms are polynomial time.

\paragraph{Polynomiality of Deterministic Case with Constant $b$.}
We now first observe that in the deterministic case with constant $b$, it is simple to solve the problem exactly in polynomial time, simply by going through all at most $\binom{n^2 - m}{b} \le n^{2b}$ possible sets of non-edges $F$, computing the deterministic set $S_F$ that the greedy algorithm outputs for maximizing $\sigma(\cdot, F)$, and checking what is the value $\tau_F = \min_{C\in \CCC} \sigma_C(S_F, F)$. Then return the set $F$ that achieves the maximum $\tau_F$. Although this seems trivial, we notice that such an approach cannot work for \FIMAL, for which we showed that the problem remains \NP-hard in the deterministic case even if $b=1$, see~\cref{thm: FIMAE inapx b1}.
\begin{observation}\label{obs: FIMALG det}
    There is a polynomial time algorithm to compute an optimal solution to \FIMALG in the deterministic case when $b$ is constant.
\end{observation}

\paragraph{Hardness.}
In the language of parameterized complexity, Observation~\ref{obs: FIMALG det} shows that the deterministic \FIMALG problem belongs to the class $\XP$ when parameterized by $b$. A natural question is therefore whether there exists an $\FPT$ algorithm that solves or approximates \FIMALG in deterministic instances. In fact, already \cref{thm: FIMAE inapx k1} answers negatively to this question as the proof shows a polynomial-time reduction from the \textsc{Set Cover} problem to the deterministic case of \FIMALG in which $b$ is equal to the size of a set cover $\kappa$. As \textsc{Set Cover} is $\operatorname{W}[2]$-hard w.r.t.\ $\kappa$, \FIMALG does not admit an $\FPT$ algorithm w.r.t.\ $b$, even in the deterministic case, unless $\operatorname{W}[2] =\FPT$. Moreover, under the same condition, no parameterized $\alpha$-approximation algorithm exists since the optimum of a \FIMALG instance is strictly positive if and only if there exists a set cover of size $\kappa$.

A natural next question is what happens for general $b$, but with $k=1$. The problem remains hard in this case. Consider the instance constructed in the reduction in~\cref{thm: FIMAE inapx k1}. As $k=1$ and as the instance is deterministic, it is clear that the greedy algorithm, for any set $F\subseteq \bar E$ of non-edges, simply computes a maximizing set of cardinality 1. Hence the following statement can be shown in the same way as in the proof of~\cref{thm: FIMAE inapx k1}: there exists a set cover $\SSS$ of size at most $\kappa$ if and only if there exists a set of non-edges $F \subseteq \bar E$ with $|F|\le b$, such that $\min_{C\in \CCC}\E_{S\sim p(F,k)}\sigma_C(S, F)]\ge 1$. This yields the following corollary to~\cref{thm: FIMAE inapx k1}.
\begin{corollary} \label{cor: FIMAEg inapx}
	For any $\alpha \in (0, 1]$, it is \NP-hard to approximate the \FIMALG problem to within a factor of $\alpha$, even in the deterministic case and if $k=1$.
\end{corollary}
As mentioned above, we will see below that \FIMALG for general constant $b$ turns out to be arbitrarily well approximable. To prove this, we first turn back to the question of approximating $\E_{S\sim p(F,k)}[\sigma_C(S, F)]$ for a fixed $F$.

\paragraph{Approximating $p(F, k)$.}
As mentioned above, we do not assume access to $p(F, k)$, instead we show that, using the greedy algorithm in an oracle fashion, we can approximate $\E_{S\sim p(F,k)}[\sigma_C(S, F)]$ arbitrarily well using a Hoeffding bound. We first recall that already $\sigma_C$ cannot be evaluated exactly but has to be approximated using $\poly(n,\eps^{-1})$ many samples of live-edge graphs.
\begin{restatable}{lemma}{hoeffdinggreedy}\label{lem: hoeffding greedy}
    Given an instance $(G, \CCC, b, k)$ of \FIMALG with constant $b$, one can in $\poly(n,m,\eps^{-1})$ time, compute functions $f_C$ such that, w.h.p., 
    \(
        | f_C(F) - \E_{S\sim p(F, k)} [\sigma_C(S, F)] | 
        \le \eps
    \)
    for all $C\in \CCC$ and $F\subseteq \bar E$ with $|F|\le b$. Here $m=|C|$.
\end{restatable}
\begin{proof}
	Following our considerations on approximation in the preliminaries, we assume to have access to approximations $\tilde \sigma_C$ of $\sigma_C$ for all $C\in \CCC$ such that, w.h.p., $|\sigma_C(S, F) - \tilde \sigma_C(S, F)| \le \eps/2$ for all $C\in \CCC$, $S\subseteq V$, and $F\subseteq \bar E$ with $|F|\le b$. Such approximations can, e.g., be computed as in Lemma 4.1 of the paper by Becker et al.~\cite{BeckerDGG22}. Concluding from the bound on $T$ there, this can be done in $\poly(n, m, \eps^{-1})$ time. We can now, for every  $F\subseteq \bar E$ with $|F|\le b$, call the greedy algorithm $N=\Omega(\eps^{-2} \log(n m))$ times and obtain sets $S_1,\ldots, S_N$ of size $k$. For every $C\in \CCC$, define $f_C(F):= \frac{1}{N}\sum_{i=1}^N \tilde \sigma_C(S_i, F)$ and $\bar f_C(F):= \frac{1}{N}\sum_{i=1}^N \sigma_C(S_i, F)$. Then using a Hoeffding bound, see, e.g., Theorem 4.12 in the book by Mitzenmacher and Upfal~\cite{mitzenmacher2017probability}, it holds that $\Pr[|\bar f_C(F) - \E_{S\sim p(F, k)} [\sigma_C(S, F)]| \ge \eps/2] \le (n m)^{-\Omega(1)}$. After applying a union bound, we obtain that w.h.p., we have  $|\bar f_C(F) - \E_{S\sim p(F, k)} [\sigma_C(S, F)]| \le \eps/2$ for all $C\in \CCC$ and $F\subseteq \bar E$ with $|F|\le b$. Hence, w.h.p., 
 	\begin{align*}
		\big| f_C(F) - \E_{S\sim p(F, k)} [&\sigma_C(S, F)] \big| \le \big|\bar f_C(F) - \E_{S\sim p(F, k)} [\sigma_C(S, F)] \big|\\ 
		& + \big| f_C(F) - \bar f_C(F) \big|
		\le \eps. \qedhere
	\end{align*}
\end{proof}

\paragraph{General Approximation for Constant $b$.}
The above lemma enables us to provide a polynomial time algorithm for \FIMALG when $b$ is constant that finds a set $F\subseteq \bar E$ that is $\eps$-close to optimal (in an additive sense) w.h.p. After proving the above lemma, the idea is simple and similar to the deterministic case: Again, go through all at most $n^{2b}$ possible sets of non-edges $F$, compute $\eps/2$-approximations $(f_C(F))_{C\in \CCC}$ as in~\cref{lem: hoeffding greedy}, and return the set with maximum value $\tau_F = \min_{C\in \CCC} f_C(F)$. This set is an additive $\eps$-approximation of the maximizing set $F^*$ (using the approximation guarantee once for $F$ and once for $F^*$).
\begin{lemma}\label{lem: FIMALG apx}
    Let $\eps\in (0,1)$, there is a polynomial time algorithm to compute an additive $\eps$-approximation to the optimal solution of \FIMALG when $b$ is constant.
\end{lemma}

\paragraph{Practical Algorithms.}
For the case with general budget $b$, recall that the problem is in-approximable unless $\PP=\NP$ according to~\cref{cor: FIMAEg inapx}. We still propose several algorithms in this paragraph that perform well in practice as we will show later on. All our algorithms are of a greedy flavour and based on restricting to the evaluation of increments of non-edges that seem promising to improve fairness. 
In the following, we describe the proposed methods.
\begin{description}
	\item[\algo{grdy\_al}.] The algorithm that, starting with $F=\emptyset$, in $b$ iterations, chooses the non-edge $e$ into $F$ that maximizes the increment $\min_C \E_{S\sim p(F, k)}[\sigma(S, F \cup \{e\})] - \min_C \E_{S\sim p(F, k)}[\sigma(S, F)]$. For efficiency we restrict to evaluate only non-edges that are (1) incident to $S_p$, the union over all sets with positive support in $p(F, k)$, and (2) are inter-community edges. Note that at the beginning of each iteration, we recompute $p(F,k)$ as $F$ changes. 
	\item[\algo{to\_minC\_infl}.] The algorithm that, starting from the empty set $F=\emptyset$, adds the non-edge $e=(u,v)\in \bar E\setminus F$ to $F$ that connects a node from $S_p$ with a node that maximizes $f(e):=\Pr_{S\sim p(F,k)}[u\in S]\cdot w_{e} \cdot \E_{S\sim p(F, k)}[\sigma_{\bar C}(S\cup \{v\}, F)]$, where $\bar C$ is the community of minimum coverage. We refer the reader to the pseudo-code in Algorithm~\ref{algorithm:to_minC_max_inflC}.
	The rationale being to choose the non-edge that connects a seed node with a node that has large influence in the community $\bar C$ taking into account both the probability that $u$ is a seed and the edge weight $w_{e}$.
	\item[\algo{to\_minC\_min}.] The algorithm that, starting from the empty set, adds a non-edge to the node $\bar v$ with minimum probability of being reached in the community that currently suffers the smallest community coverage. Among all these non-edges we choose the non-edge $(u, \bar v)$ that maximizes the product $\Pr_{S\sim q}[u\in S] \cdot w_{(u, \bar{v})}$. The pseudo-code is given in Algorithm~\ref{algorithm:to-minC-min-v}.
\end{description}
We highlight two techniques that we use speed up our implementations: (1) a pruning technique for \algo{grdy\_al}: Let $\delta$ denote the best increment of an edge that we have seen so far. Before evaluating the exact increment of a non-edge $e=(u, v)\in A\setminus F$, we compute an upper bound on the increment achievable by $e$ via evaluating the expected community coverages $\E_{S\sim p(F,k)}[\sigma_C(\{v\}, F)]$ that would be achieved by choosing $v$ as a seed. We refer the reader to the pseudo-code in Algorithm~\ref{algorithm:GreedyEdgesSC} for further details. (2) A way to update RR sets rather than recompute them from scratch after adding edges: In all our algorithms, we change the graph by adding edges to it. As a consequence the functions $\sigma$ and $\sigma_C$ need to be approximated based on different simulations or, here, based on different RR sets. We observe however that after adding one edge, say $e=(u,v)$ to the graph, we do not need to entirely resample the RR sets, but, instead, can update and reuse them as follows. For every RR set $R$ that contains the node $v$, we update $R$ by re-starting the RR set construction from $u$ with probability $w_{(u,v)}$ and adding the resulting nodes to $R$.

\begin{algorithm}[ht]
	\caption{\algo{grdy\_al}}
	\begin{algorithmic} 
		\REQUIRE instance \(\III=(G, \CCC, b, k)\)
		\ENSURE set $F\subseteq \bar E$ with $|F|\le b$
		
		\STATE \(F\leftarrow\emptyset\)\;
		\WHILE{$|F| < b$}
		\STATE \(q \leftarrow p(F, k)\)\;
		\STATE \(\delta \leftarrow -\infty \)\;
		\STATE \(A \leftarrow \{(u,v)\in \bar E: u\in S \text{ for some }S: q_S>0 \text{ and }v\notin S \text{ for all }S:q_S>0\}\)
		\FOR{\((u, v)=e\in A \setminus F\)}
		\STATE \( \tau_C(v) \leftarrow \E_{S\sim q}[\sigma_C(S, F)] + \E_{S\sim q}[\sigma_C(\{v\}, F)], \ \text{for all}\ C \in \CCC  \)\;
		\IF{\(\min_{C\in \CCC} \{\tau_C(v)\} > \delta \)}
		\STATE \(\lambda \leftarrow \min_{C\in \CCC} \{\E_{S\sim q}[\sigma_C(S, F \cup \{e\})]\}\)\;			
		\IF{\( \lambda > \delta   \)}
		\STATE \(\delta \leftarrow \lambda\)\;
		\STATE \(\bar{e} \leftarrow e\)\;
		\ENDIF
		\ENDIF
		\ENDFOR
		\STATE $F \leftarrow F \cup \{\bar{e}\}$ 
		\ENDWHILE
		\RETURN $F$
	\end{algorithmic}
	\label{algorithm:GreedyEdgesSC}
\end{algorithm}

\begin{algorithm}[ht]
	\caption{\algo{to\_minC\_infl}}
	\begin{algorithmic} 
		\REQUIRE instance \(\III=(G, \CCC, b, k)\)
		\ENSURE set $F\subseteq \bar{E}$ with $|F|\le b$
		
		\STATE \(F\leftarrow\emptyset\)\;
		\WHILE{$|F|<b$}
		\STATE \(q \leftarrow p(F, k)\)\;
		\STATE \(\bar{C}\leftarrow \argmin_{C\in\CCC}\{\E_{S\sim q}[\sigma_C(S, F)]\}\)\;
		\STATE \(\bar{e}\leftarrow \argmax_{(u, v)\in \bar{E}\setminus F}\{ \Pr_{S\sim q}[u\in S] \cdot w_{(u, v)}\cdot\E_{S\sim q}[\sigma_{\bar{C}}(S\cup\{v\}, F)]\} \)\;
		\STATE $F \leftarrow F \cup \{\bar{e}\}$ 
		\ENDWHILE
		\RETURN $F$
	\end{algorithmic}
		\label{algorithm:to_minC_max_inflC}
\end{algorithm}

\begin{algorithm}[ht]
	\caption{\algo{to\_minC\_min}}
	\begin{algorithmic} 
		\REQUIRE instance \(\III=(G, \CCC, b, k)\)
		\ENSURE set $F\subseteq \bar{E}$ with $|F|\le b$
		
		\STATE \(F\leftarrow\emptyset\)\;
		\WHILE{$|F| < b$}
		\STATE \(q \leftarrow p(F, k)\)\;
		\STATE \( \bar{C}\leftarrow \argmin_{C\in \CCC}\{\E_{S\sim q}[\sigma_C(S, F)]\}\)\;
		\STATE \( \bar{v}\leftarrow \argmin_{v\in \bar{C}}\{\E_{S\sim q}[\sigma_v(S, F)]\} \)\;
		\STATE \( \bar{e} \leftarrow \argmax_{(u, \bar{v})\in \bar{E}\setminus F} \{\Pr_{S\sim q}[u\in S] \cdot w_{(u, \bar{v})}\} \)\;
		\STATE $F \leftarrow F \cup \{\bar{e}\}$ 
		\ENDWHILE
		\RETURN $F$
	\end{algorithmic}
		\label{algorithm:to-minC-min-v}
\end{algorithm}

\section{Experiments}\label{sec: experiments}
In this section, we report on two experiments involving the \FIMALG problem. In the first experiment, we compare the algorithms presented above in terms of quality and running time. In a second experiment, we evaluate the best performing algorithm against other fairness-tailored seeding algorithms. We show, for several settings, that already adding just a few edges can lead to a situation where purely efficiency-oriented information spreading becomes automatically fair.\footnote{The code can be downloaded from 
\url{https://github.com/sajjad-ghobadi/fair_adding_links.git}} We proceed by describing the experimental setup.

\paragraph{Experimental Setting.}
In our experiments we use random, synthetic and real world instances.
(1)~Random instances are generated using the Barabasi-Albert model connecting newly added nodes to two existing nodes.  
(2)~The synthetic instances are the ones 
used by Tsang et al.~\cite{TsangWRTZ19}. Each network consists of 500 nodes and every node is associated with some attributes (region, ethnicity, age, gender and status) that induce communities. Nodes with similar attributes are more likely to share an edge.
(3)~We use similar real world instances as Fish et al.~\cite{Fish19}. We proceed by describing the real world instances.
\texttt{Arenas}~\cite{guimera2003self} and \texttt{email-Eu-core}~\cite{LeskovecKF07} are email communication networks at the
University Rovira i Virgili (Spain) and a
large European research institution, respectively. Each user corresponds to a node and there is a directed
edge between two users if at least one message is sent between them. In \texttt{email-Eu-core}, each user belongs to one
of 42 departments that defines a community structure.
\texttt{irvine}~\cite{OpsahlP09} is a network created from an online social network
at the University of California, irvine. 
Each node corresponds to a student and the network contains a directed edge if at least one online message was sent among the students.
\texttt{youtube}~\cite{YangL15} consists of a part of the social network in Youtube. There is a node for each user and each edge represents the friendship between two users.
In Youtube, the community structure is defined by Youtube groups, where each user can define a group and
others can join. For \texttt{youtube}, we considered a connected sub-network of size 3000 using the community structure.
After removing the nodes that do not belong to any community, we consider a sub-network consisting of the first 3000 nodes reached by a BFS from a random source
node. We also remove singleton communities, thus some of the nodes
may not belong to any community. The number of communities is 1575.
\texttt{ca-GrQc} (General Relativity and Quantum Cosmology) and \texttt{ca-HepTh} (High Energy
Physics - Theory)~\cite{LeskovecKF07} are co-authorship networks for
two different categories of arXiv. Each node corresponds to an author and an undirected edge between two nodes represents that they authored a paper together.
To avoid zero probabilities in the experiments, for all the real world instances (other than \texttt{youtube}), we considered the largest weakly connected component.
The properties of all instances are summarized in Table~\ref{instanses}. 

\begin{table}[htp]\small
	\begin{center}
		\begin{tabular}{c|c|c|c} 
			Dataset & $\#$Nodes & $\#$Edges & Direction  \\
			\hline
			\texttt{Barabasi-Albert} & $200$ & $792$ & Directed\\
			\texttt{Synthetic} & $500$ & $1576$-$1697$ & Directed\\
			\texttt{email-Eu-core} & $1005$ & $25571$ & Directed\\
			\texttt{Arenas} & $1133$ & $5451$ & Directed\\
			\texttt{irvine} & $1899$ & $20296$ & Directed\\
			\texttt{youtube} & $3000$ & $29077$ & Undirected\\
			\texttt{ca-GrQc} & $5242$ & $14496$ & Undirected\\
			\texttt{ca-HepTh} & $9877$ & $25998$ & Undirected\\
		\end{tabular}
		\caption{Properties of random, synthetic and real world networks (sorted by $n$).}
		\label{instanses}
	\end{center}
\end{table}

\begin{figure}[ht]
	\centering
	\includegraphics[trim={0 3.6cm 0 0}, clip, width=0.95\linewidth]{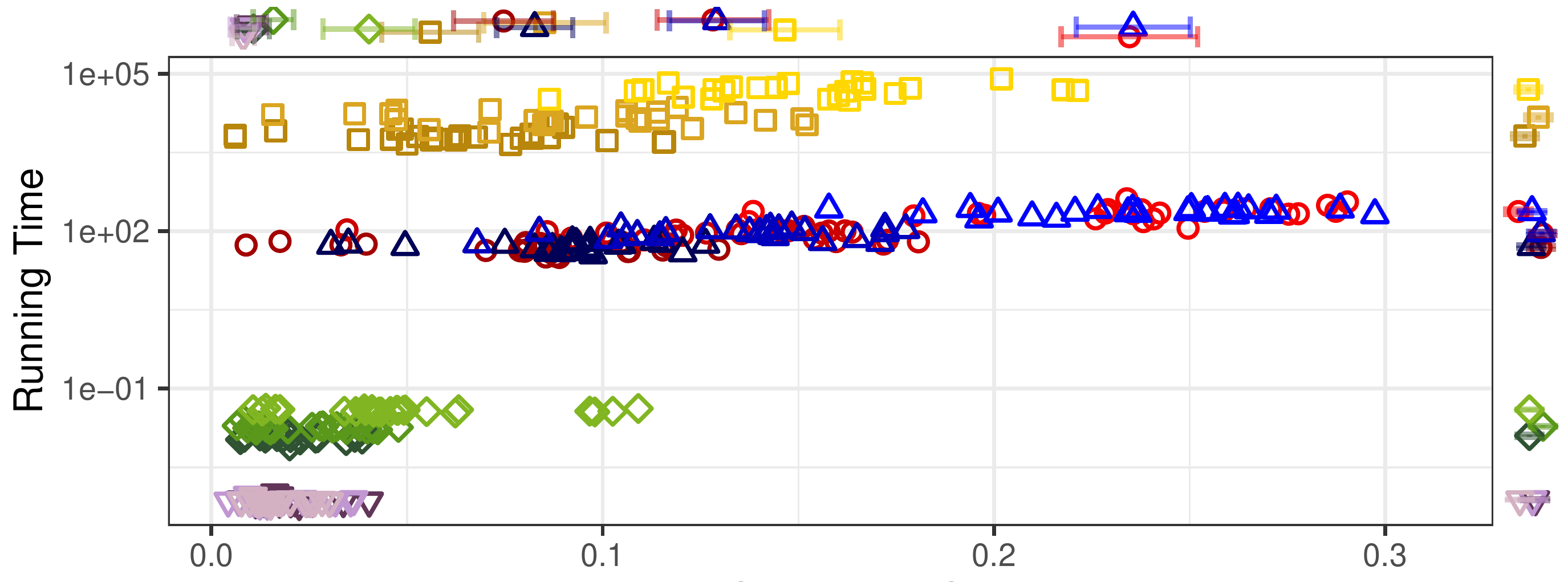} 
	\includegraphics[trim={0.2cm 0.5cm 0.2cm 0.0cm}, clip, width=0.95\linewidth]{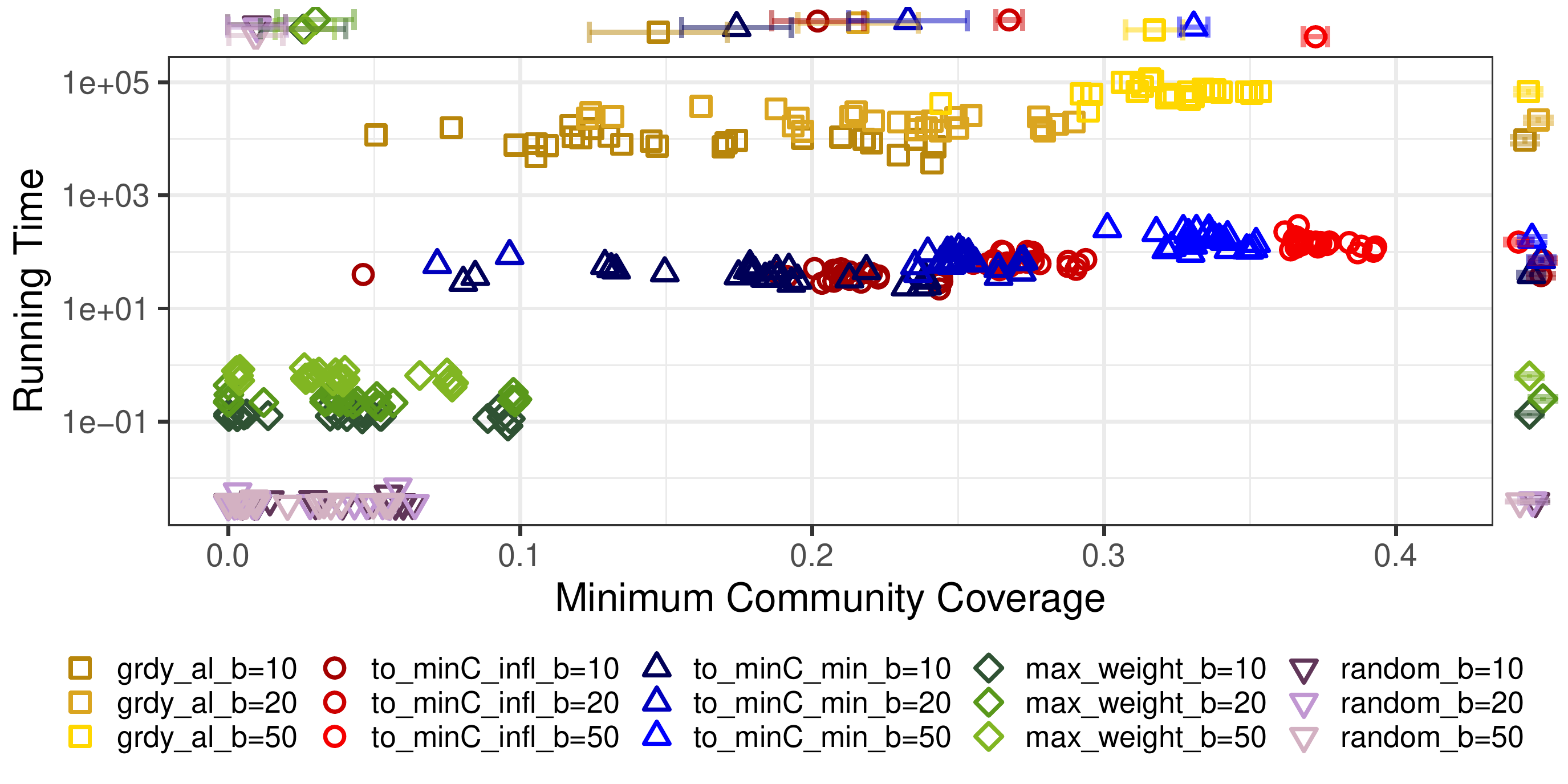} 
	\caption[Results Experiment 1: (1)~Random instances, (2)~synthetic instances]{Results Experiment 1: (1)~Random instances ($k=25$, $n=200$, singleton communities), (2)~synthetic instances ($k=25$, $n=500$, communities induced by gender and region). 
	The running time is on the logarithmic vertical axis, while the minimum community coverage is on the horizontal axis.
	}
	\label{fig: experiment 1 1}
\end{figure}

For random and synthetic instances we select edge weights uniformly at random in the interval $[0, 0.4]$, and in the interval $[0, 0.2]$ for the real world instances (other than \texttt{youtube}). For \texttt{youtube}, we choose the edge weights uniformly at random in the interval $[0, 0.1]$. We choose the non-edge weights uniformly at random from the interval $[0, 1]$.
We consider different community structures: (1)~Singleton communities: each node has its own community. (2)~BFS communities: for every $i \in [m]$, we generate a community $C_i$ of size $n/m$ using a breadth first search from a random source node (we continue this process if the size of a community is less than $n/m$). (3)~Community structures given for the synthetic networks and some of the real world networks.

We repeat each algorithm 5 times per graph. For random and synthetic instances, we average in addition over 5 graphs, thus resulting in 25 runs per algorithm. The error-bars in our plots represent 95-\% confidence intervals.
All experiments were executed
on a compute server running Ubuntu 16.04.5 LTS with 24
Intel(R) Xeon(R) CPU E5-2643 3.40GHz cores and a total
of 128 GB RAM.

We use the TIM implementation for IM by Tang, Xiao, and Shi~\cite{TangXS14} in order to implement the greedy algorithm for IM.
We note that our algorithms, \algo{grdy\_im}, and \algo{mult\_weight} are implemented in \texttt{C++} (and were compiled with g++ 7.5.0),
while \algo{moso}, \algo{grdy\_maxmin} and \algo{myopic} are implemented in \texttt{python} (we use
\texttt{python} 3.7.6 for executing the code). 
For the final evaluation of the algorithms implemented in \texttt{python}, we use a constant number of
$100$ live-edge graphs for simulating the information spread (this is a common approach in the literature \cite{Fish19, FarnadiBG20,BeckerDGG22}), while for the \texttt{C++} implementations we use the number of RR sets generated by the TIM implementation. For the final evaluation of ex-ante values, we set $\epsilon = \delta = 0.1$ to obtain an additive $\eps$-approximation with probability at least $1-\delta$. 

\begin{figure}[htp]
	\centering
	\includegraphics[trim={0 3.6cm 0 0}, clip, width=0.95\linewidth]{ 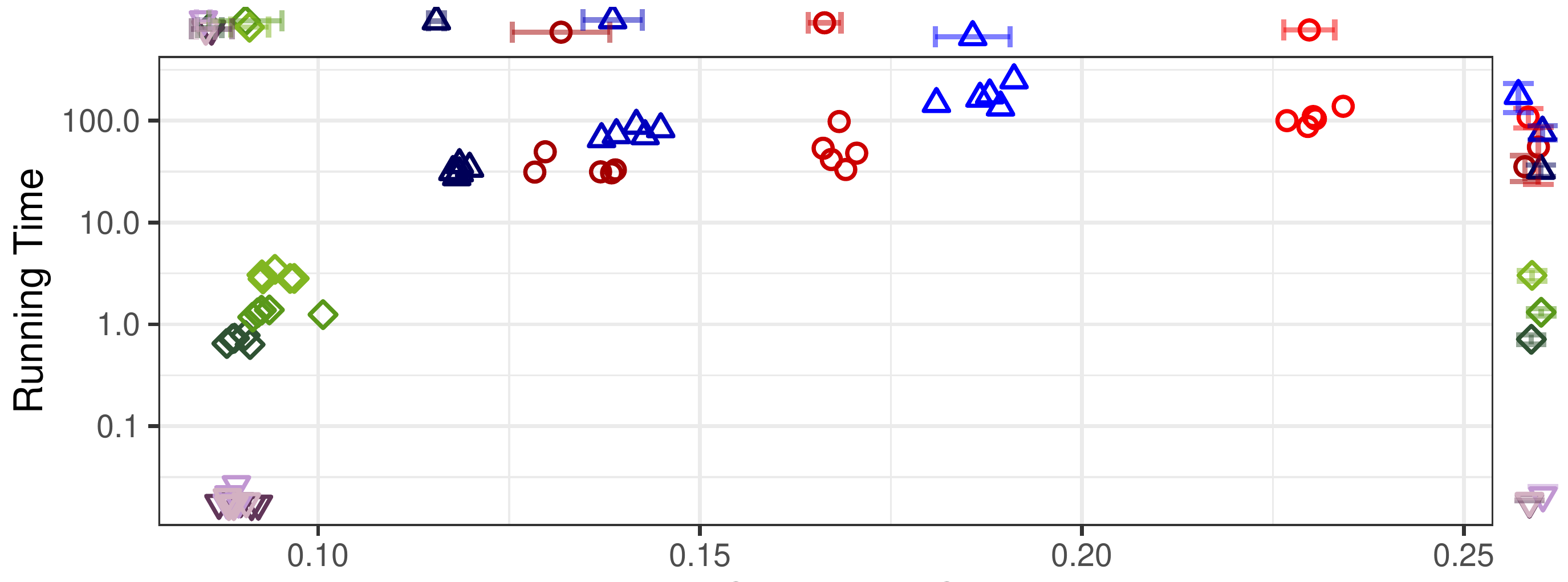} 
	\includegraphics[trim={0 3.6cm 0 0}, clip, width=0.95\linewidth]{ 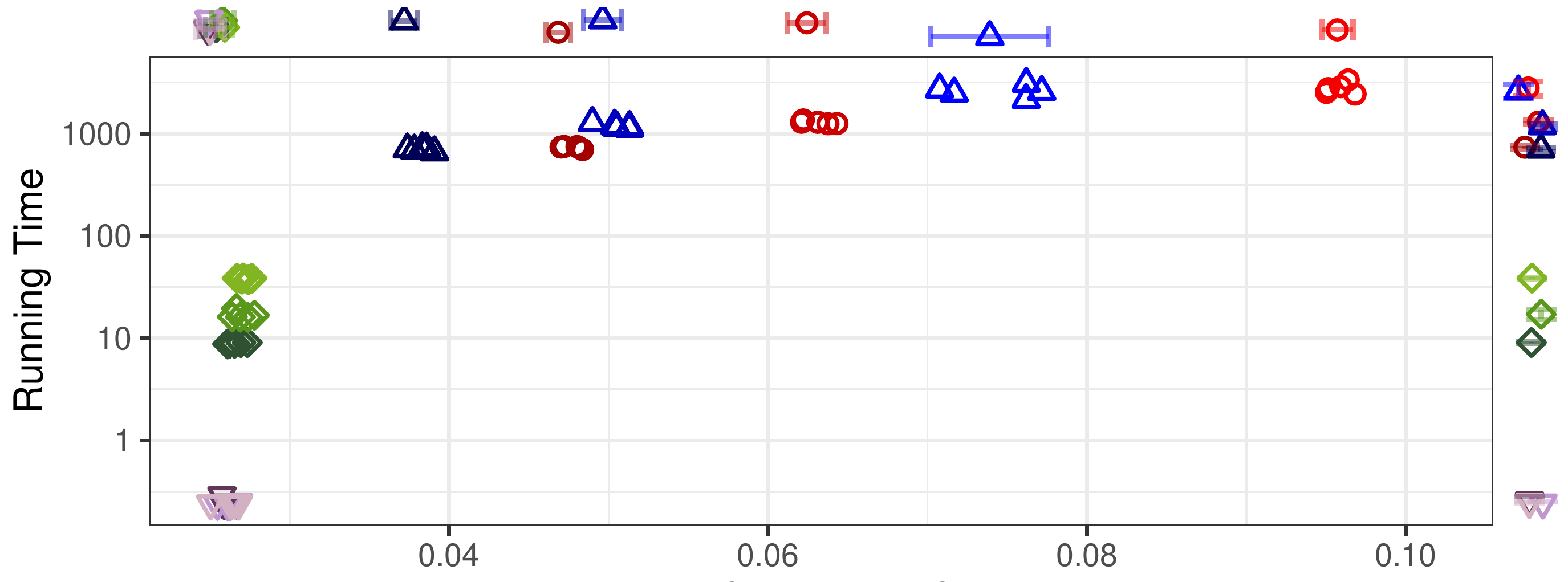} 
	\includegraphics[width=0.95\linewidth]{ 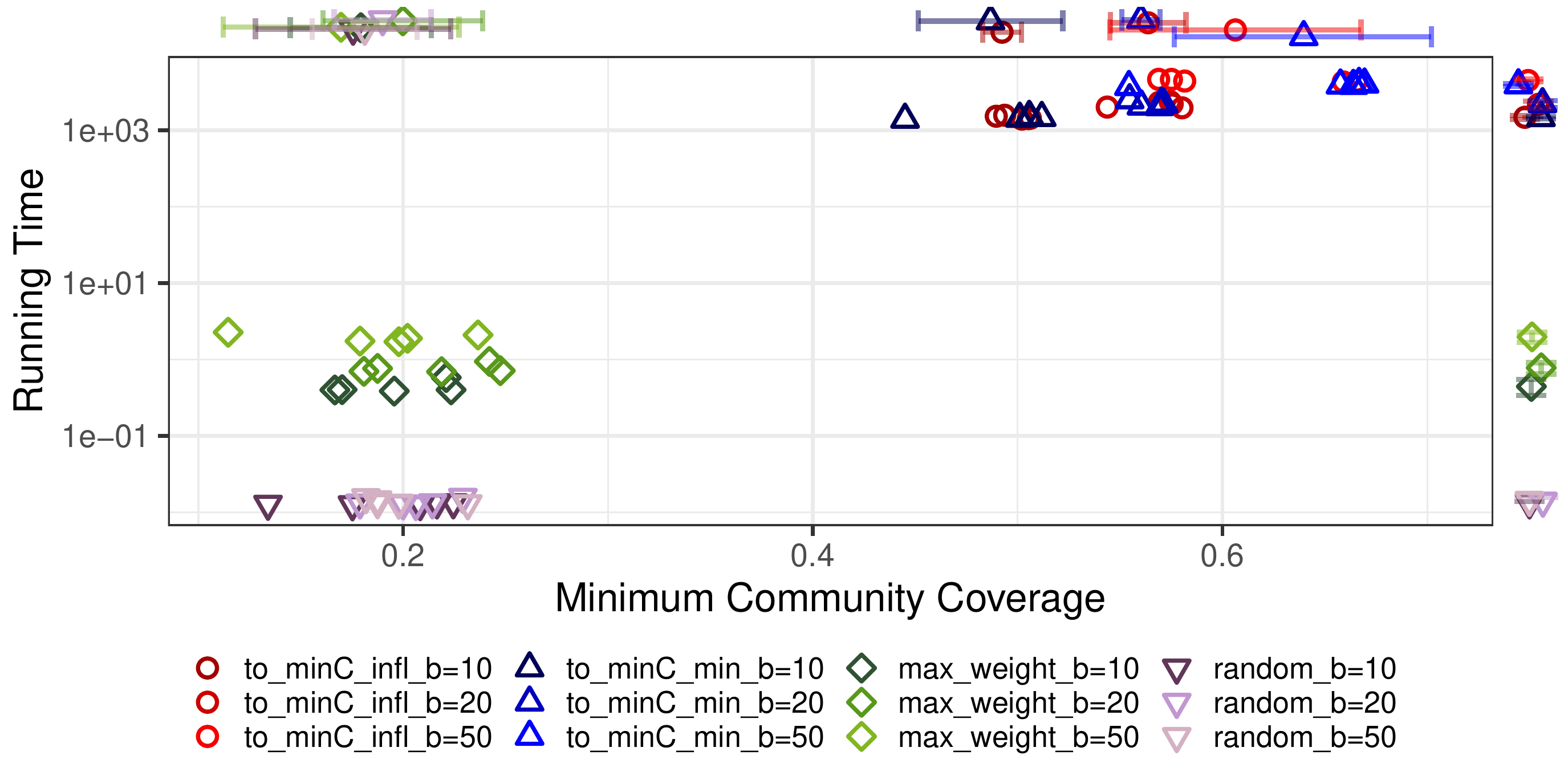} 
	\caption[Results Experiment 1:
		(1)~\textsf{Arenas} with BFS communities,
		(2)~\textsf{ca-GrQc} with BFS communities,
		(3)~\textsf{email-Eu-core} with real communities]{Results Experiment 1:
		(1)~\textsf{Arenas} with BFS communities ($m=10$), $k=20$,
		(2)~\textsf{ca-GrQc} with BFS communities ($m=10$), $k=20$,
		(3)~\textsf{email-Eu-core} with real communities, $k=20$. 
		Again, the running time is on the logarithmic vertical axis, while the minimum community coverage is on the horizontal axis.
	}
	\label{fig: experiment 1 2}
\end{figure}

\paragraph{Experiment 1.}
In addition to the three algorithms described in Section~\ref{sec: fimalg}, we evaluate the following two base lines: \algo{random}: the algorithm that chooses $b$ non-edges uniformly at random, and \algo{max\_weight}: the algorithm that chooses the $b$ non-edges of maximal weight. The results can be found in Figure~\ref{fig: experiment 1 1} for the random and synthetic instances.
We observe that, despite the pruning approach described above, \algo{grdy\_al}'s running time is the worst. Furthermore, the fairness that it achieves is worse than the one of \algo{to\_minC\_infl}. We thus exclude \algo{grdy\_al} from further experiments. \algo{random} and \algo{max\_weight} are fastest but the fairness achieved by
them is very poor. 

In Figure~\ref{fig: experiment 1 2}, we can see the results for the real world instances \textsf{Arenas}, \textsf{ca-GrQc} and \textsf{email-Eu-core}.
We observe that the running times of both algorithms \algo{to\_minC\_infl} and \algo{to\_minC\_min} are comparable, while \algo{to\_minC\_infl} achieves better values of fairness. We thus choose \algo{to\_minC\_infl} as the best performing algorithm as a result of this experiment.

\begin{figure}[ht]
	\centering
	\includegraphics[trim={0 3.4cm 0 0}, clip, width=0.95\linewidth]{ 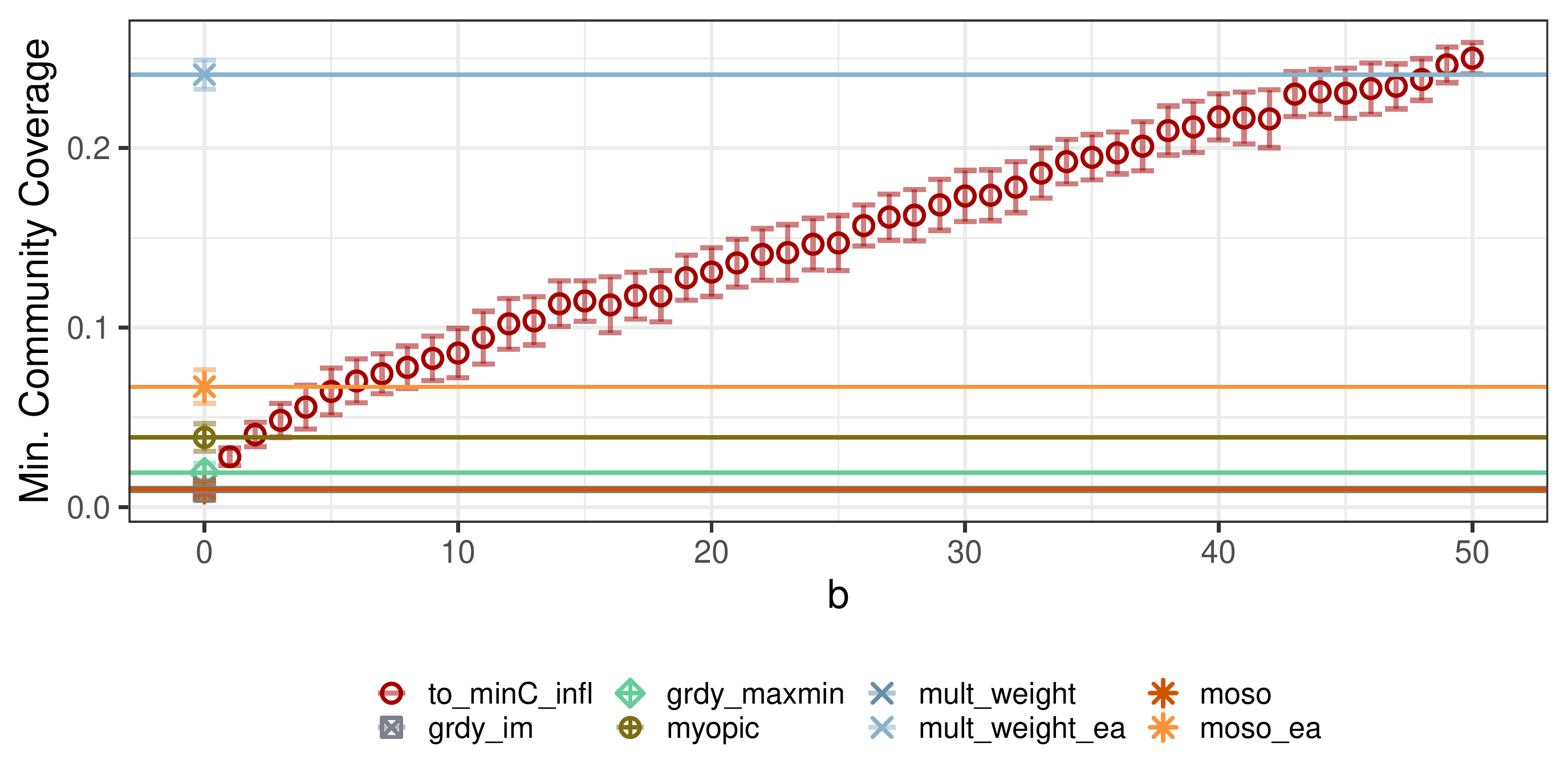} 
	\includegraphics[ width=0.95\linewidth]{ 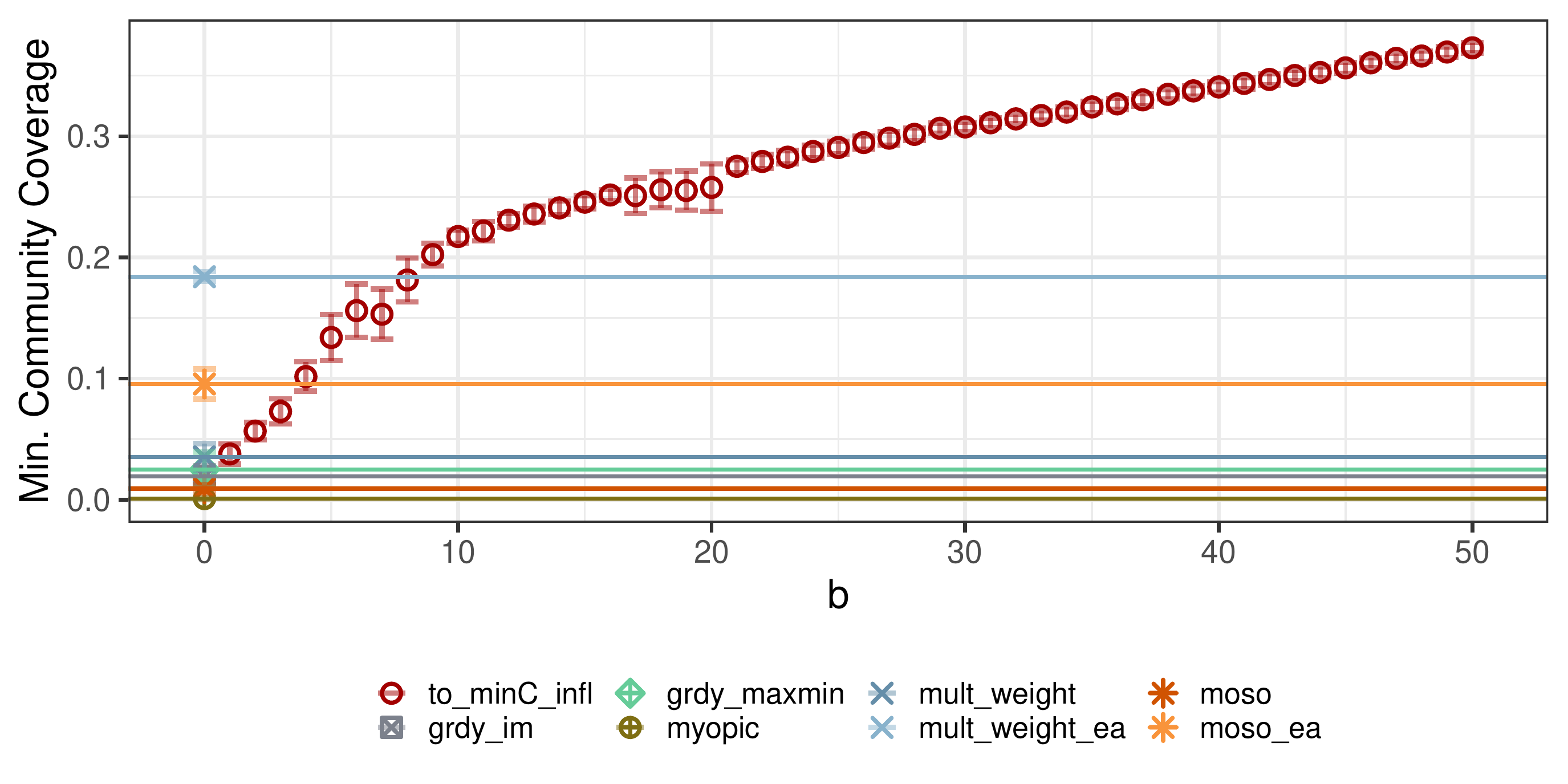} 
	\caption[Results Experiment 2: (1) Random instances, 
	(2)~synthetic instances]{Results Experiment 2: (1) Random instances ($k=25$, $n=200$, singleton communities), 
	(2)~synthetic instances ($k=25$, $n=500$, communities induced by gender and region), minimum community coverage on the vertical axis, $b$ on the horizontal axis.
 }
	\label{fig: experiment 2 1}
\end{figure}

\paragraph{Experiment 2.}
The goal of the second experiment is to analyze how many links we need to add in order to make the standard greedy algorithm for IM satisfy similar or better fairness guarantees than fairness-tailored algorithms.
To this end, we compare our method \algo{to\_minC\_infl} with the following competitors: \algo{grdy\_im}, the standard greedy algorithm for IM (we use the implementation of Tang, Xiao, and Shi~\cite{TangXS14}) serves mainly as a baseline; \algo{grdy\_maxmin}, the greedy algorithm that iteratively selects $k$ seed nodes to maximize the minimum community coverage; \algo{myopic}, a simple heuristic proposed by Fish et al.~\cite{Fish19} that iteratively chooses the node with minimum probability of being reached as seed; \algo{mult\_weight}, the multiplicative weights routine for the set-based problem of Becker et al.~\cite{BeckerDGG22}; \algo{moso}, a multi-objective submodular optimization approach proposed by Tsang et al.~\cite{TsangWRTZ19} (we choose gurobi as solver~\cite{gurobi}). 

We note that the algorithms \algo{mult\_weight} and \algo{moso} are designed to compute distributions over seed sets and nodes, respectively, and thus they can be used to obtain both ex-ante and ex-post fairness guarantees. We defer the reader to the work of Becker et al.~\cite{BeckerDGG22} for details regarding probabilistic seeding and ex-ante guarantees. Hence, for these two algorithms we include both there ex-post and ex-ante values in our evaluations. It is worth pointing out that is much easier (especially in settings with many communities) to achieve good values ex-ante rather than ex-post.

We show the results for the random and synthetic instances in Figure~\ref{fig: experiment 2 1}.
Already for small values of $b$, i.e., after adding just a few edges, our algorithm surpass all ex-post fairness values of the competitors. Even better and maybe surprisingly, our algorithm also achieves ex-post values higher than the ex-ante values of \algo{mult\_weight} and \algo{moso}.
We exclude the algorithms \algo{grdy\_maxmin} and \algo{moso} from experiments with the real world instance as they perform the worst in terms of running time.
We turn to the real world instances, see Figure~\ref{fig: experiment 2 2}, on which we evaluate our algorithm for three fixed values of $b=10,20,50$. We observe that by adding only 10 edges, the fairness values obtained by our algorithm dominate over the ex-post fairness values achieved by the competitors. 
We also observe that after adding only 50 edges, the fairness values of our method are larger than (or comparable to) the ex-ante fairness values achieved by \algo{mult\_weight}, on all instances.

\begin{figure}[ht]
	\centering
	\includegraphics[width=0.95\linewidth]{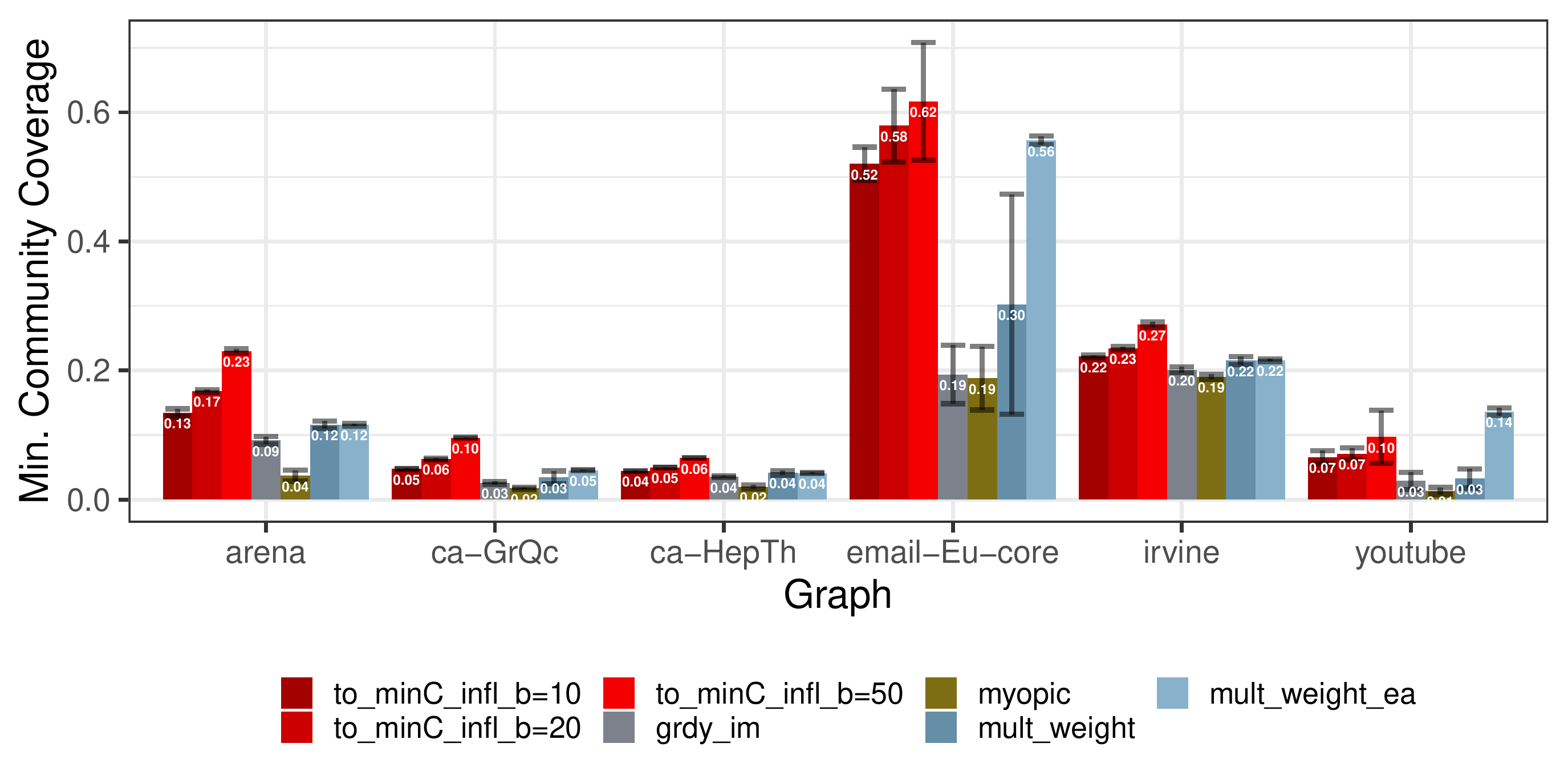} 
	\caption{Results Experiment 2: Real world graphs with BFS communities ($m=10$) for \texttt{arena}, \texttt{ca-GrQc}, \texttt{ca-HepTh}, \texttt{irvine} and real communities for \texttt{email-Eu-core} and \texttt{youtube}, $k=20$,  minimum community coverage on the vertical axis, different instances on the horizontal.}
	\label{fig: experiment 2 2}
\end{figure}

\section{Conclusion}
We studied two optimization problems with the goal of adding links to a social network such as to make purely efficiency-oriented information spreading automatically fair. In the first problem \FIMAL, our goal is to add at most $b$ non-edges $F$ to the graph such that the minimum community coverage $\sigma_C(S, F)$ is maximized w.r.t.\ maximizing sets $S$ of size at most $k$ to spread information.
We showed several hardness and hardness of approximation results for \FIMAL. Maybe most importantly, the decision version of \FIMAL is $\Sigma_2^p$-hard even in the deterministic case and remains \NP-hard even if $b=1$ or $k=1$ (in the latter case even to approximate within any factor). We thus proposed to study a second optimization problem \FIMALG that entails to add at most $b$ non-edges $F$ to the graph such that the minimum expected community coverage is maximized when information is spread using the greedy algorithm for influence maximization. As we observed, also this problem remains \NP-hard to approximate to within any factor if $k=1$. On the other hand, in contrast to \FIMAL, \FIMALG becomes polynomial time $-\eps$-approximable if $b$ is a constant. We then proposed several heuristics for \FIMALG and evaluated them in an experimental study. Lastly, we conducted an experiment showing that the greedy algorithm for IM achieves similar or even better levels of fairness than fairness-tailored algorithms already after adding a few edges proposed by our algorithm.

\bibliography{references}



\end{document}